\documentclass[11pt, oneside]{article} 
\pdfoutput=1

\usepackage{jheppub}
\usepackage{amsmath,amssymb,amsfonts,amsthm}
\usepackage{wrapfig}
\usepackage{overpic}
\usepackage{color}
\usepackage{booktabs}
\definecolor{darkblue}{rgb}{0.1,0.1,.7}
\usepackage[]{amsmath}
\usepackage{graphicx} 

\usepackage[]{latexsym}
\usepackage{amscd}
\usepackage[all,cmtip]{xy}
\usepackage{mathrsfs}
\usepackage{bbold}
\usepackage{subfigure}
\usepackage[margin=10pt,font=small,labelfont=bf]{caption}
\usepackage{subcaption} 
\usepackage{simplewick}
\usepackage{changepage}
\usepackage[]{algorithm2e}
\usepackage{booktabs,multirow}
\usepackage{float}

\usepackage[OT2,OT1]{fontenc}
\usepackage{braket}
\usepackage{tikz}
\usepackage{pgfplots}
\pgfplotsset{compat=1.10}
\usepgfplotslibrary{fillbetween}
\usetikzlibrary{patterns}

\newcommand{\abs}[1]{\left\lvert#1\right\rvert}

\newtheorem{theorem}{Theorem}[section]
\newtheorem{lemma}[theorem]{Lemma}
\newtheorem{corollary}[theorem]{Corollary}
\newtheorem{proposition}[theorem]{Proposition}
\newtheorem{conjecture}[theorem]{Conjecture}
\theoremstyle{remark}
\newtheorem{remark}[theorem]{Remark}

\makeatletter
\def\@fpheader{\ }
\makeatother

\title{A universal inequality on the unitary 2D CFT partition function}
\author{Indranil Dey$^{1}$, Sridip Pal$^{2}$, Jiaxin Qiao$^{3}$}
\affiliation{$^1$Department of Theoretical Physics, Tata Institute of Fundamental Research,
Colaba, Mumbai, India, 400005 \\$^2$Walter Burke Institute for Theoretical Physics,  California Institute of Technology,  Pasadena, CA, USA  \\  $^3$ Laboratory for Theoretical Fundamental Physics, Institute of Physics, École Polytechnique Fédérale de Lausanne (EPFL), CH-1015 Lausanne, Switzerland\\ }

\abstract{We prove the conjecture proposed by Hartman,  Keller and Stoica (HKS) \cite{Hartman:2014oaa}: the  grand-canonical free energy of a unitary 2D CFT with a sparse spectrum below the scaling dimension $\frac{c}{12}+\epsilon$ and below the twist $\frac{c}{12}$ is universal in the large $c$ limit for all $\beta_L\beta_R \neq 4\pi^2$.  

\vspace{0.3cm}
The technique of the proof allows us to derive a one-parameter (with parameter $\alpha\in(0,1]$) family of universal inequalities on the unitary 2D CFT partition function with general central charge $c\geqslant 0$,  using analytical modular bootstrap.  We derive an iterative equation for the domain of validity of the inequality on the $(\beta_L,\beta_R)$ plane.  The infinite iteration of this equation gives the boundary of maximal-validity domain, which depends on the parameter $\alpha$ in the inequality.
\vspace{0.1cm}

In the $c \to \infty$ limit, with the additional assumption of a sparse spectrum below the scaling dimension $\frac{c}{12} + \epsilon$ and the twist $\frac{\alpha c}{12}$ (with $\alpha \in (0,1]$ fixed), our inequality shows that the grand-canonical free energy exhibits a universal large $c$ behavior in the maximal-validity domain. This domain, however, does not cover the entire $(\beta_L, \beta_R)$ plane, except in the case of $\alpha = 1$. For $\alpha = 1$, this proves the conjecture proposed by HKS \cite{Hartman:2014oaa}, and for $\alpha < 1$, it quantifies how sparseness in twist affects the regime of universality. Furthermore, this implies a precise lower bound on the temperature of near-extremal BTZ black holes, above which we can trust the black hole thermodynamics.

}

\begin{document}
	\begin{flushright}
CALT-TH 2024-039
\end{flushright}
\maketitle

\section{Introduction \& summary}
Fundamental physical principles such as locality and crossing symmetry impose nonperturbative constraints on the physical observables in conformal field theories.  In two dimensions,  the partition function of a CFT on a torus is modular invariant.  A famous consequence of modular invariance is the Cardy formula,  a universal formula for the spectral density of high-energy states in a unitary modular invariant CFT,  derived in \cite{Cardy:1986ie} and established rigorously in \cite{Mukhametzhanov:2019pzy,Mukhametzhanov:2020swe,Pal:2019zzr}.  Assuming a twist gap in the spectrum of Virasoro primaries,\footnote{In this paper, the terminology ``twist" refers to $\Delta-J$, where $\Delta$ is the scaling dimension and $J$ is the spin. In 2D CFT, we have $\Delta=h+\bar{h}$ and $J=\abs{h-\bar{h}}$, so the twist is given by $2\min\left\{h,\bar{h}\right\}$.}  the Cardy formula has been extended to the regime of small twist,  large spin for $c>1$ unitary CFTs \cite{Benjamin:2019stq,Pal:2022vqc,Pal:2023cgk}.  Further consequences of modular invariance are explored in \cite{Kraus:2016nwo,Kraus:2017kyl,Kraus:2018pax,Cardy:2017qhl,Das:2017vej,Das:2017cnv,Kusuki:2018wpa,Collier:2018exn,Hikida:2018khg,Romero-Bermudez:2018dim,Brehm:2018ipf,Benjamin:2020zbs} and in \cite{Ganguly:2019ksp,Pal:2019yhz,Pal:2020wwd,Das:2020uax} with mathematical rigor.  The 2D CFT partition function has also been analyzed using the tools of harmonic analysis on the modular group $\mathrm{SL}(2,\mathbb{Z})$.  This avatar of the analytic modular bootstrap has been explored in \cite{Benjamin:2021ygh,Benjamin:2022pnx,DiUbaldo:2023qli,Haehl:2023tkr,Haehl:2023xys,Haehl:2023mhf}. 

Modular invariance has been used along with linear programming to bound the scaling dimension of lowest non-vacuum operators as function of the central charge,  see \cite{Hellerman:2009bu} for the original argument,  \cite{Collier:2016cls} for a revisitation,  \cite{Afkhami-Jeddi:2019zci} for the state of the art numerical result and \cite{Hartman:2019pcd} for the state of the art analytical result  in this direction. 

In this paper,  we use the analytical modular bootstrap method to explore the universality of the free energy for a 2D CFT in a grand-canonical ensemble{\color{black}, which depends on the left and right inverse temperatures, $\beta_{L}$ and $\beta_{R}$.\footnote{$\beta_{L}$ and $\beta_{R}$ are related to the thermodynamic quantitites, the temperature $T$ and the angular potential $\mu$, by $\beta_{L}=(1+\mu)/T$ and $\beta_{R}=(1-\mu)/T$. The standard thermodynamic formula is given by $dE=TdS+\mu dJ$, where $E$ is the energy, $S$ is the entropy and $J$ is the angular momentum.} We would like to study the following question:
\begin{itemize}
	\item For a unitary, modular invariant 2D CFT, suppose we approximate the free energy using the vacuum term of the partition function. When is the relative error small?
\end{itemize}
To be precise, we decompose the free energy into two parts\footnote{In this paper, the terminology ``free energy" differs from the standard one by a minus sign.}
\begin{equation}\label{FE:lowT}
	\begin{split}
		\log Z(\beta_{L},\beta_{R})=\frac{c}{24}(\beta_{L}+\beta_{R})+\mathcal{E}(\beta_{L},\beta_{R}),\quad(\beta_{L}\beta_{R}\geqslant4\pi^2).
	\end{split}
\end{equation}
Here $Z(\beta_{L},\beta_{R})$ is the partition function. In the right-hand side, the first term comes from the vacuum contribution $\log e^{\frac{c}{24}(\beta_{L}+\beta_{R})}$, and the second term $\mathcal{E}(\beta_{L},\beta_{R})$ is the error term which we would like to estimate. The reason we only consider the regime $\beta_{L}\beta_{R}\geqslant4\pi^2$ is that the torus partition function satisfies modular invariance, $Z(\beta_{L},\beta_{R})=Z(4\pi^2/\beta_{L},4\pi^2/\beta_{R})$, which tells us that the partition function in the regime $\beta_{L}\beta_{R}\leqslant4\pi^2$ can be written as
\begin{equation}\label{FE:highT}
	\begin{split}
		\log Z(\beta_{L},\beta_{R})=\frac{c}{24}\left(\frac{4\pi^2}{\beta_{L}}+\frac{4\pi^2}{\beta_{R}}\right)+\mathcal{E}\left(\frac{4\pi^2}{\beta_{L}},\frac{4\pi^2}{\beta_{R}}\right),\quad(0<\beta_{L}\beta_{R}\leqslant4\pi^2).
	\end{split}
\end{equation}
By comparing the vacuum terms in \eqref{FE:lowT} and \eqref{FE:highT}, we see that the one in \eqref{FE:lowT} is larger in the low-temperature regime ($\beta_{L}\beta_{R}\geqslant4\pi^2$) and the one in \eqref{FE:highT} is larger in the high-temperature regime ($\beta_{L}\beta_{R}\leqslant4\pi^2$), and the transition happens at the self-dual line $\beta_{L}\beta_{R}=4\pi^2$.

In a fixed theory, the error term $\mathcal{E}(\beta_{L},\beta_{R})$ is known to be small compared to the vacuum term in some limit. The simplest example is when $\beta_{L}=\beta_{R}=\beta\gg 1/\Delta_{\rm gap}$, where $\Delta_{\rm gap}$ is the scaling dimension gap of the theory. In this case, the error term is controlled by:
\begin{equation}\label{estimate:equaltemp}
	\begin{split}
		\mathcal{E}(\beta,\beta)={\rm O}\left(e^{-\Delta_{\rm gap}\beta}\right),\quad\beta\rightarrow\infty.
	\end{split}
\end{equation}
Here we would like to point out that the coefficient in front of the exponential decay $e^{-\Delta_{\rm gap}\beta}$ is fixed once the theory is fixed. This is how Cardy formula for free energy works in this case \cite{Cardy:1986ie}.

However, it may happen that the theory depends on several parameters and the central charge $c$ is tunable in terms of these parameters. This situation occurs in both high energy and condensed matter systems, such as 3D AdS gravity \cite{Banados:1992wn, Banados:1992gq, Maldacena:1998bw, Maloney:2007ud,Hartman:2019pcd}, Liouville theory \cite{Polyakov:1981rd,Seiberg:1990eb,Teschner:2001rv} and critical loop models \cite{Nienhuis:1984wm,diFrancesco:1987qf,Nivesvivat:2023kfp}. An interesting limit is to take $c\rightarrow\infty$ with $\beta_{L}$ and $\beta_{R}$ fixed. This limit is important for the black hole microstate counting \cite{Strominger:1996sh}, but it goes beyond the scope of the standard Cardy formula. One may still want to derive some bounds on the error term $\mathcal{E}(\beta_{L},\beta_{R})$, which we may expect is small compared to the linear growth in $c$, at least in some regime of $\beta_{L}$ and $\beta_{R}$. 

The main point of this paper is to derive some universal bounds on $\mathcal{E}(\beta_{L},\beta_{R})$ to answer the above question. While the authors' exploration started with an aim of establishing such a universality in the $c\to\infty$ limit and make statements about the partition function of 3D gravity in AdS$_3$,  the main result of this paper is proven for any fixed unitary, modular invariant CFT, without assuming holography.  Hence,  we decide to organize the remaining introduction section in two parts.  The first part is generic,  valid for all unitary CFTs with $c\geqslant 0$.  In the second part,  we discuss the holographic implication, which follows from the first part with additional sparseness conditions on the low lying spectrum. A special case of these sparseness conditions was postulated in \cite{Hartman:2014oaa}.

\paragraph{\textbf{An estimate for the free energy for all unitary 2D CFTs with $c\geqslant0$:\vspace{0.2cm}\\}}
In a holomorphic CFT,  the light spectrum $\Delta\leqslant c/12$,  completely determines the partition function due to Rademacher expansion \cite{rademacher1937convergent}.  However,  given that such an expansion does not apply for a generic CFT,  it is meaningful to ask how much mileage we can have just by knowing the light spectrum.  This question has been expounded on in reference \cite{Kaidi:2020ecu}, and from the point of view of harmonic analysis in \cite{Benjamin:2021ygh} assuming some spectral condition and well definedness of modular completion of the light part of the partition function. 

Let us suppose,  we know the spectrum of a unitary 2D CFT for scaling dimension below $c/12+\epsilon$ with small $\epsilon>0$.  We define
\begin{equation}\label{def:ZwidehatL}
\widehat{Z}_{L}(\beta):=\sum_{\Delta\leqslant c/12+\epsilon} e^{-\beta\Delta}\,.
\end{equation}
Can we estimate the canonical free energy, which is $\log Z(\beta_{L},\beta_{R})$ with $\beta_{L}=\beta_{R}=\beta$, as a function of $\beta$?  This is answered in \cite{Hartman:2014oaa} by the following estimate\footnote{The upper bound is true for any $c\geqslant0$ modular invariant, unitary CFTs. The lower bound requires an extra condition that the theory has a normalizable vacuum.} 
\begin{equation}
\log Z(\beta)=\frac{c}{12}\beta+  \varepsilon(\beta)\,,\quad 0\leqslant \varepsilon(\beta) \leqslant \log\left(\frac{\widehat{Z}_{L}(\beta)}{1-e^{-\left(\beta-\frac{4\pi^2}{\beta}\right)\epsilon}}\right)\,,\quad \beta>2\pi\,.
\end{equation}
The next question is: can we go further into the mixed-temperature regime, where $\beta_{L}$ and $\beta_{R}$ are not necessarily equal? Using unitarity, one can easily extend the above estimate to the regime $\beta_{L},\beta_{R}>2\pi$, where the error term in \eqref{FE:lowT} is bounded as follows
\begin{equation}
	\begin{split}
		0\leqslant \mathcal{E}(\beta_{L},\beta_{R})\leqslant\varepsilon(\beta),\quad\beta=\min\left\{\beta_{L},\beta_{R}\right\}, \\
	\end{split}
\end{equation}
where $\varepsilon(\beta)$ is the same as above. Using modular invariance, one gets a similar estimate in the regime $\beta_{L},\beta_{R}<2\pi$.

Now can we go even further, e.g. into the regime $\beta_{L}\leqslant2\pi<\beta_{R}$? Extending the above estimate into this regime is the main result of this paper.

Let us suppose,  we know the spectrum of a unitary 2D CFT with scaling dimension below $c/12+\epsilon$ with small $\epsilon>0$ as well as the spectrum with twist below $\alpha c/12$, where $\alpha\in(0,1]$. We define the following quantity couting the low-twist contribution to the partition function:
\begin{equation}\label{intro:def:ZtildeLH}
	\begin{split}
		\tilde{Z}_L(\alpha;\beta_{L},\beta_{R})&:=\sum\limits_{\text{min}(h,\bar{h})\leqslant\frac{\alpha c}{24}}n_{h,\bar{h}}\ e^{- h\beta_L-\bar{h}\beta_{R}}. \\
	\end{split}
\end{equation}
Given this, can we estimate the error $\mathcal{E}(\beta_{L},\beta_{R})$ as a function of $\beta_L$ and $\beta_R$? 

In this paper,  we use modular invariance to derive such an estimate on $\mathcal{E}(\beta_{L},\beta_{R})$, which follows from a variant of theorem \ref{theorem:logZ} stated below. Our estimate will be valid for $(\beta_{L},\beta_{R})$ in an open domain $\mathcal{D}_\alpha\subset\{\beta_{L}\beta_{R}>4\pi^2\}$,  given by 
	\begin{equation}\label{def:Dalpha0}
	\begin{split}
		\mathcal{D}_\alpha&:=\left\{(\beta_{L},\beta_{R})\ \Big{|}\ \beta_{R}>2\pi,\ \beta_{L}>\phi_\alpha(\beta_{R});\ {\rm or}\ \beta_{L}>2\pi,\ \beta_{R}>\phi_\alpha(\beta_{L})\right\}, \\
		\phi_\alpha(\beta)&:=\max\left\{\frac{4\pi^2}{\beta},\ \frac{1}{2} \left(\alpha(4\pi - \beta) + \sqrt{\alpha^2(4\pi - \beta)^2 + 16\pi^2 (1 - \alpha)}\right)\right\}\,.
	\end{split}
\end{equation}
For each $(\beta_{L},\beta_{R})\in \mathcal{D}_\alpha$, we algorithmically determine a finite collection of points:
$$(\beta_L^{(1)},\beta_R^{(1)}),\ (\beta_L^{(2)},\beta_R^{(2)}),\ \ldots\ (\beta_L^{(N)},\beta_R^{(N)})\quad (N<\infty),$$
staying in the following regime
\begin{equation}
	\begin{split}
		\beta_{L}^{(i)}\,\beta_{R}^{(i)}&>4\pi^2\ (i=1,2,\ldots,N-1),\quad \beta_{L}^{(N)}, \beta_{R}^{(N)}>2\pi\,. \\
	\end{split}
\end{equation}
\begin{figure}[!ht]
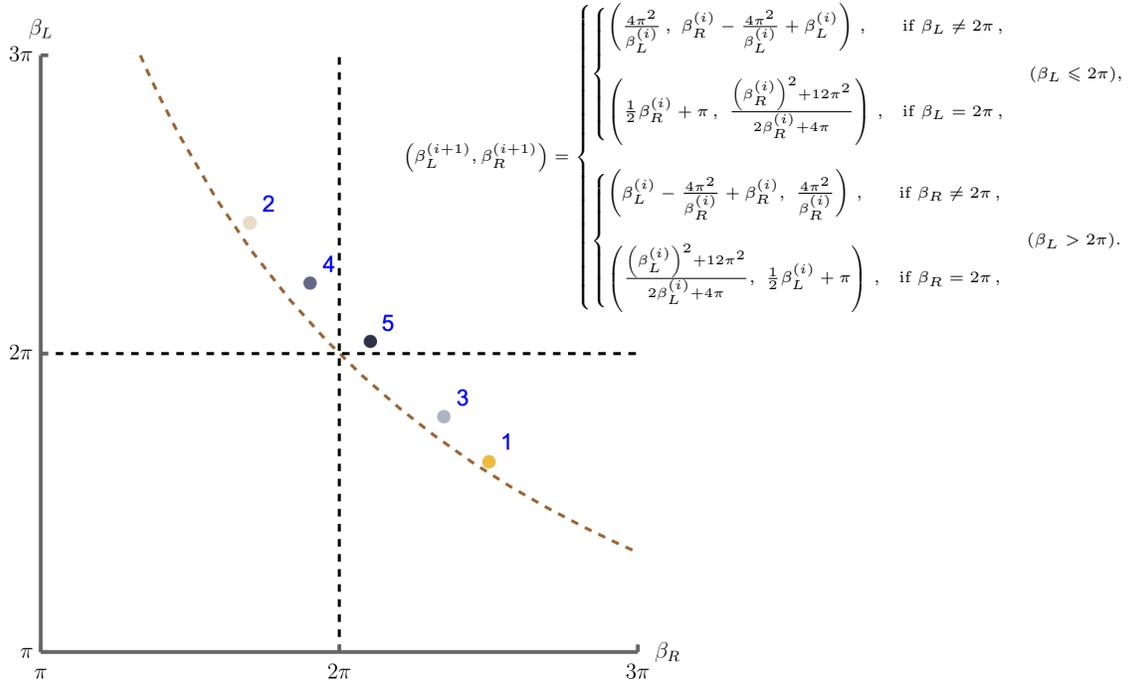

	\begin{overpic}[width=0.6\textwidth,tics=10]{updates\_algo.pdf}
		\put (59, 76) {\fontsize{6pt}{7pt}\selectfont $\left(\beta_{L}^{(i+1)},\beta_{R}^{(i+1)}\right)=\begin{cases}\begin{cases} \left(\frac{4 \pi ^2}{\beta_{L}^{(i)}}\,,\ \beta_{R}^{(i)}-\frac{4 \pi ^2}{\beta_{L}^{(i)}}+\beta_L^{(i)}\right)\,,&\text{if}\ \beta_L\neq 2\pi\,,\\
					\\
					\left(\frac{1}{2}\beta_{R}^{(i)}+\pi\,,\ \frac{\left(\beta_{R}^{(i)}\right)^2+12 \pi^2}{2 \beta_{R}^{(i)}+4 \pi } \right)\,,&\text{if}\ \beta_L=2\pi\,,
				\end{cases}\quad(\beta_L\leqslant 2\pi),\\
				\\
				\begin{cases} \left(\beta_{L}^{(i)}-\frac{4 \pi ^2}{\beta_{R}^{(i)}}+\beta_R^{(i)},\ \frac{4 \pi ^2}{\beta_{R}^{(i)}}\right)\,,&\text{if}\ \beta_R\neq 2\pi\,,\\
					\\
					\left(\frac{\left(\beta_{L}^{(i)}\right)^2+12 \pi^2}{2 \beta_{L}^{(i)}+4 \pi },\ \frac{1}{2}\beta_{L}^{(i)}+\pi \right)\,,&\text{if}\ \beta_R=2\pi\,,
				\end{cases}\quad (\beta_L>2\pi).
			\end{cases}$}
	\end{overpic}
	\caption{An example,  showing the set of points $(\beta_L^{(i)},\beta_R^{(i)})$, obtained using the algorithm, appearing in fig.\  \ref{fig:algorithmbeta},  for $\alpha=1$, starting from an initial choice of $(\beta_L,\beta_R)=(\frac{41\pi}{25},\frac{5\pi}{2})$. We begin with $(\beta_L^{(1)},\beta_R^{(1)})=(\frac{41\pi}{25},\frac{5\pi}{2})$ and the recursion terminates at $N=5$.  The numbers near the points denote at which step they are reached, added by $1$.}
	\label{fig:0}
\end{figure}The algorithm begins with an initial pair $(\beta_L, \beta_R) \in \mathcal{D_\alpha}$. At each step, these values are updated, and the algorithm terminates at step $N$ if $ \beta_{L}^{(N)}, \beta_{R}^{(N)} > 2\pi$.   The algorithm is explained in figure \ref{fig:algorithmbeta}.  {See figure \ref{fig:0} as well for an example with $\alpha=1$.}

A non-trivial result that we establish is that $\mathcal{D_\alpha}$ is the maximal domain in which the iteration terminates after a finite number of steps. At each step, we utilize the modular invariance of the partition function to derive the following key inequality (as explained later in section \ref{sec:prooflemma} in the proof \footnote{In section \ref{sec:prooflemma}, $g(x,y) = e^{-\frac{2\pi c}{24}(x+y)}Z(2\pi x,2\pi y)$ and $g_0(x,y) = \tilde{Z}_L(2\pi x, 2\pi y)$} of lemma \ref{lemma:functionbound}):
	\begin{equation}
		e^{-\frac{c}{24}(\beta_{L}^{(i)} + \beta_{R}^{(i)})}Z(\beta_{L}^{(i)}, \beta_{R}^{(i)}) - \tilde{Z}_L\left(\alpha; \beta_{L}^{(i)}, \beta_{R}^{(i)}\right) \leqslant e^{-\frac{c}{24}(\beta_{L}^{(i+1)} + \beta_{R}^{(i+1)})}Z(\beta_{L}^{(i+1)}, \beta_{R}^{(i+1)}).
	\end{equation}
	Primarily, using this step repeatedly, we show that for any $c\geqslant0$ unitary, modular invariant 2D CFT with a normalizable vaccum, the error term $\mathcal{E}(\beta_{L},\beta_{R})$ satisfies the following bound
\begin{equation}\label{modifiedTh}
	\begin{split}
		0\leqslant\mathcal{E}(\beta_{L},\beta_{R})&\leqslant\log\left[\sum_{i=1}^{N}\tilde{Z}_L\left(\alpha;\beta_{L}^{(i)},\beta_{R}^{(i)}\right)+\frac{\widehat{Z}_{L}(\beta)}{1-e^{-\left(\beta-\frac{4\pi^2}{\beta}\right)\epsilon}}\right], \\
		\beta&:=\text{Min}\left(\beta_{L}^{(N)}, \beta_{R}^{(N)}\right),
	\end{split}
\end{equation}
The key point to emphasize is that for $\alpha=1$,  the domain $\mathcal{D}_1$ is the same as $\{\beta_{L}\beta_{R}>4\pi^2\}$. Therefore, $\mathcal{D}_1$ and its $S$-dual domain $\mathcal{D}_1'=\{\beta_{L}\beta_{R}<4\pi^2\}$ encompass the whole $(\beta_L,\beta_R)$ plane, except for the self-dual line $\beta_L\beta_R=4\pi^2$. So we have an estimate on the error term $\mathcal{E}(\beta_{L},\beta_{R})$ according to \eqref{FE:lowT} and \eqref{FE:highT}, aside from the self-dual line. As we tune $\alpha$,  the domain of validity of the bound reduces.  For more details,  see the theorem\!~\ref{theorem:logZ} and the remarks that follow. 

Eq.\,\eqref{modifiedTh} and theorem \ref{theorem:logZ} do not make use of Virasoro symmetry. By incorporating the full Virasoro symmetry, the result is strengthened in a refined theorem, which applies to the Virasoro-primary partition function $Z^{\rm Vir}(\beta_{L},\beta_{R}) := Z(\beta_{L},\beta_{R}) \eta(\beta_{L}) \eta(\beta_R)$, counting only Virasoro primaries (see theorem \ref{theorem:logZVir}). The main differences are: 1) the analog of $\hat{Z}_L$ is now constructed from Virasoro primaries with $\text{min}(h,\bar{h}) \leqslant \alpha \frac{c-1}{24}$, in contrast to all states with $\text{min}(h,\bar{h}) \leqslant \alpha \frac{c}{24}$; and 2) the bound involves additional factors of $\big(\tfrac{\beta^{(i)}_{L}\beta^{(i)}_{R}}{4\pi^2}\big)^{1/2}$, which arise from the modular covariance properties of $Z^{\rm Vir}(\beta_{L},\beta_{R})$.

\begin{figure}[!ht]
	\centering
	\includegraphics[scale=0.16]{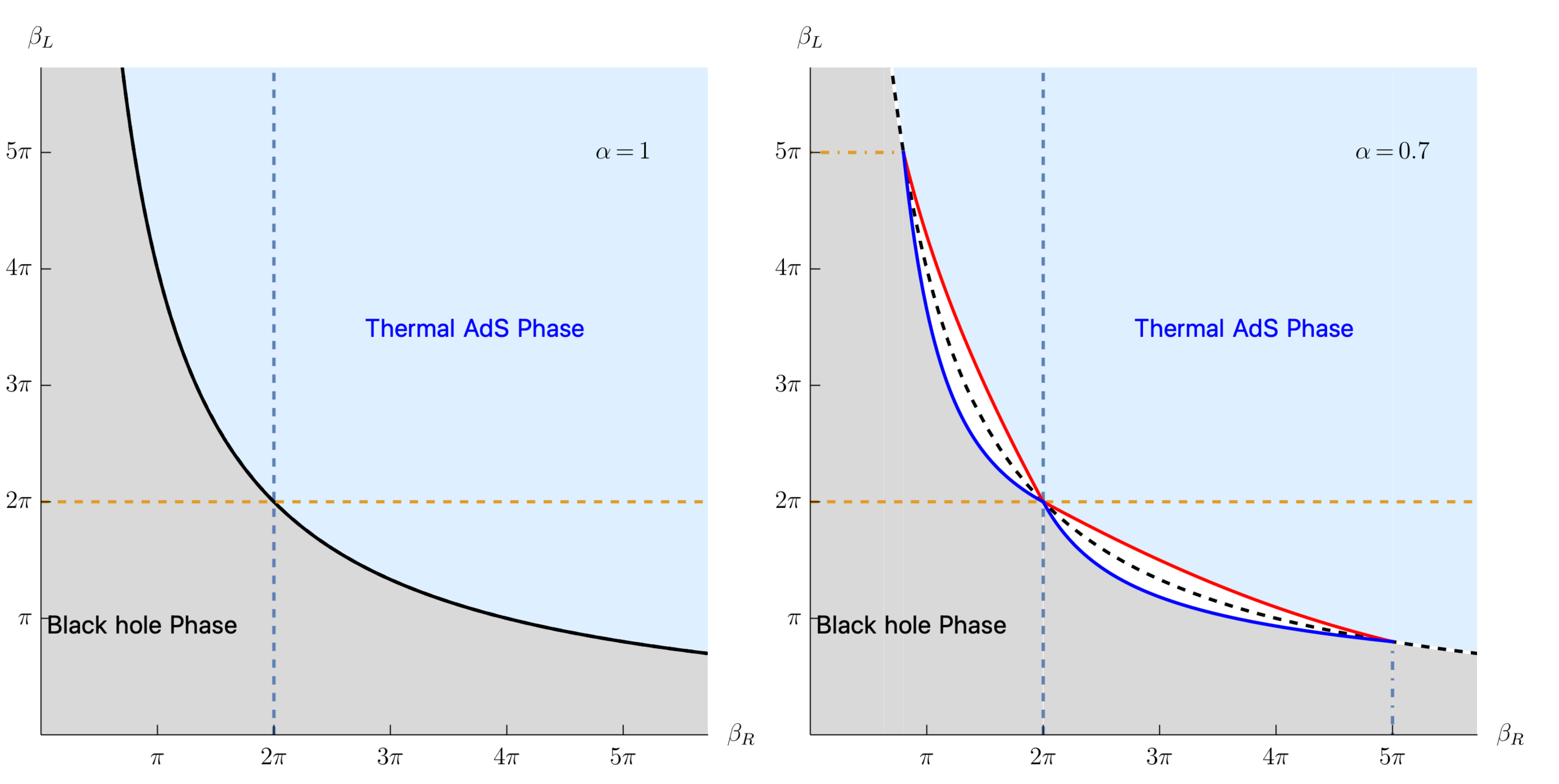}
	\caption{The large-$c$ phase diagrams corresponding to $\alpha=1$ and $\alpha=0.7$.  The grey shaded region is dominated by black hole while the blue shaded region is dominated by thermal AdS.  For $\alpha=1$,  the union of the two regions covers the whole plane except for the hyperbola $\beta_L\beta_R=4\pi^2$,  depicted as the {\bf black} solid curve in the figure on the left.
	For $\alpha<1$ (the figure on the right),  white region,  where we do not have any universality,  starts to emerge.  The boundary of the white region is analytically given by the union of curves: $\frac{\beta_L}{2\pi}=f^{(2)}_\alpha\!(\frac{\beta_R}{2\pi})$, $\frac{\beta_R}{2\pi}=f^{(2)}_\alpha\!(\frac{\beta_L}{2\pi})$ (the {\color{red}\bf red} curves) and $\frac{2\pi}{\beta_L}=f^{(2)}_\alpha\!(\frac{2\pi}{\beta_R})$,  $\frac{2\pi}{\beta_R}= f^{(2)}_\alpha\!(\frac{2\pi}{\beta_L})$ (the {\color{blue}\bf blue} curves), where $f^{(2)}_\alpha$ is defined in \eqref{twosolutions}.   For $\alpha>1/2$,  we show that 
	 the white region must cap off at the intersection point of $\beta_L\beta_R=4\pi^2$ and $\text{Max}(\beta_L,\beta_R)=\frac{2\pi}{2\alpha-1}$, which is $5\pi$ for $\alpha=0.7$ and diverges as $\alpha\to 1/2$ from above. For $\alpha=0.7$,  we have shown this feature with dotdashed line at $\text{Max}(\beta_L,\beta_R)=5\pi$.  The {\bf black} dotted curve in the figure on the right is $\beta_L\beta_R=4\pi^2$. }
	\label{fig:1}
\end{figure}

\paragraph{\textbf{Holographic implication:\vspace{0.2cm}\\}}

Conformal field theories with large central charge and sparse low lying spectra are believed to be dual to weakly coupled theories of AdS gravity,  including Einstein gravity.  The duality posits that physical observables in such CFTs should exhbit universality.  In the seminal paper \cite{Hartman:2014oaa},  Hartman,  Keller and Stoica (henceforth called HKS) showed that for a modular invariant 2D CFT with large central charge and sparse energy spectrum below $E=\Delta-c/12=\epsilon$ with a small positive $\epsilon$ i.e.\ assuming $\hat{Z}_L(\beta)=O(1)$ for $\beta>2\pi$,  a universal Cardy formula for the density of states with energy $E\sim c$ holds.   It is an extended regime of validity of the usual Cardy formula \cite{Cardy:1986ie} and expected from the semiclassical physics of black holes in AdS$_3$.   In particular,  assuming the energy sparseness condition,  the free energy of the CFT is shown to be given by the that of the black hole for $\beta<2\pi$ and of the thermal AdS for $\beta>2\pi$.  

Upon turning on temperatures conjugate to both the left movers and right movers,  HKS explored  the phase diagram (henceforth called mixed-temperature phase diagram) of the free energy on the $(\beta_L,\beta_R)$ plane.  Assuming the twist spectrum is sparse below the twist $\frac{c}{12}$ i.e.\ assuming $\tilde{Z}_L=O(1)$ for $\beta_L\beta_R>4\pi^2$,  HKS conjectured: the free energy is universal in the large $c$ limit for all $\beta_L\beta_R \neq 4\pi^2$.

Using the bound \eqref{modifiedTh} with $\alpha=1$, we show that the universal regime for the large-$c$ free energy is indeed the entire $(\beta_{L},\beta_{R})$-plane except for the self-dual line $\beta_{L}\beta_{R}=4\pi^2$. This proves the HKS conjecture. The details of the argument is presented in section \ref{subsec:HKS}.

We also study the large-$c$ behavior of the CFT free energy under a weaker assumption of sparseness in twist i.e.\ assuming  that the spectrum is sparse below the twist $\frac{\alpha c}{12}$ (with $\alpha>0$ fixed) for various values of $\alpha \in (0,1]$.\footnote{The problem is much simpler for the case of $\alpha>1$. There, the conclusion is the same as $\alpha=1$.}  The free energy has universal large-$c$ behavior for $(\beta_{L},\beta_{R})$ in the domain of validity of \eqref{modifiedTh} and its image of $S$-modular transformation. As examples,  we have plotted the phase diagram corresponding to $\alpha=1,0.7, 0.5,0.3$,  depicted in the fig.\ \ref{fig:1} and \ref{fig:2}.  See  section \ref{subsec:HKSM} for more details.

Using holography,  our result further implies a precise lower bound on near extremal BTZ black hole’s temperature at which we can trust the black hole thermodynamics.  See the conclusion section.

\begin{figure}[!ht]
	\centering
	\includegraphics[scale=0.17]{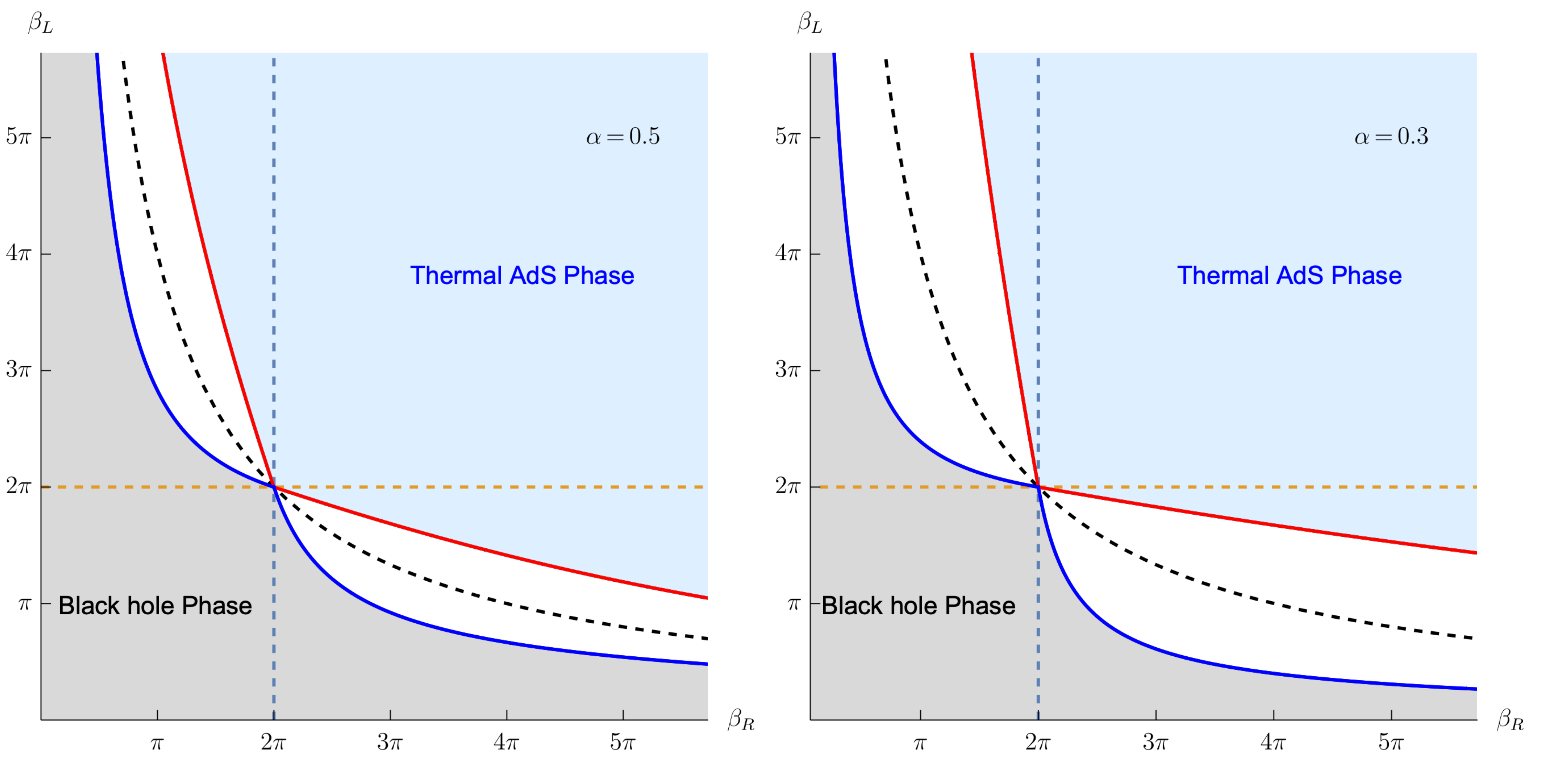}
	\caption{The large-$c$ phase diagrams corresponding to $\alpha=0.5$ and $\alpha=0.3$.  The grey shaded region is dominated by black hole while the blue shaded region is dominated by thermal AdS.  The free energy is not universal in white region.  The boundary of these white regions are analytically given by the union of curves:  $\frac{\beta_L}{2\pi}=f^{(2)}_\alpha\!(\frac{\beta_R}{2\pi})$, $\frac{\beta_R}{2\pi}=f^{(2)}_\alpha\!(\frac{\beta_L}{2\pi})$ and   $\frac{2\pi}{\beta_L}=f^{(2)}_\alpha\!(\frac{2\pi}{\beta_R})$,  $\frac{2\pi}{\beta_R}= f^{(2)}_\alpha\!(\frac{2\pi}{\beta_L})$,  where $f^{(2)}_\alpha$ is defined in \eqref{twosolutions}. For $\alpha\leq 1/2$,  the white region extends to infinity.  
	The black dotted curve is $\beta_L\beta_R=4\pi^2$}
	\label{fig:2}
\end{figure}

Given that there are infinitely many heavy states near the line \(h = \frac{c - 1}{24}\) and the line \(\bar{h} = \frac{c - 1}{24}\) in a 2D CFT with \(c > 1\) and a non-zero twist gap \cite{Collier:2016cls,Afkhami-Jeddi:2017idc,Benjamin:2019stq,Pal:2022vqc}, one might wonder whether our result requires imposing a sparseness condition on an infinite amount of data. However, the strengthened version (namely, theorem \ref{theorem:logZVir}) allows us to relax the twist sparseness condition: it only requires that the Virasoro primaries with twist below \( \frac{\alpha(c-1)}{12}\) (as opposed to all states with twist below \( \frac{\alpha c}{12}\)) are sparse. Yet, this relaxation still enables us to prove the HKS conjecture and arrive at the phase diagrams in figures \ref{fig:1} and \ref{fig:2}.

The organization of the paper is as follows.  In section \ref{sec:MR},  we present the main results.  In particular,  the main theorem \ref{theorem:logZ} is presented in section \ref{subsec:MT}.  The theorem follows from a lemma \ref{lemma:functionbound},  stated in section \ref{section:lemma}.  In section \ref{subsec:HKS},  we present the proof of HKS conjecture as a consequence of the main theorem,  followed by an extension in section \ref{subsec:HKSM}.  In section \ref{sec:prooflemma},  we prove the lemma \ref{lemma:functionbound}.  In section \ref{section:vir}, we refined our results using Virasoro symmetry and proved a strengthened theorem \ref{theorem:logZVir} along with its holographic implications.

We end this section presenting a logical flowchart to guide the readers:
\begin{equation*}
\underset{ \text{Proved in }\S\ref{sec:prooflemma}}{\text{Lemma \ref{lemma:functionbound}}}\quad \overset{\S\ref{section:lemma}}{\Longrightarrow}\quad \underset{\S\ref{subsec:MT}}{\text{Theorem \ref{theorem:logZ}}} \quad\Longrightarrow\quad \begin{cases}
\underset{\S\ref{subsec:HKS}}{\text{The proof of HKS conjecture}}\\
\underset{\S\ref{subsec:HKSM}}{\text{Beyond HKS, varying sparseness in twist}}
\end{cases}
\end{equation*}

\section{Main results}\label{sec:MR}
\subsection{CFT torus partition function, main theorem}\label{subsec:MT}
We begin with $c\geqslant0$ unitary, modular invariant 2D CFTs. The CFT torus partition function is defined by
\begin{equation}\label{def:Z}
	Z (\beta_L,\beta_R) \equiv \text{Tr}_{\mathcal{H}_{\text{CFT}}} \left( e^{- \beta_L
		\left( L_0 - \frac{c}{24} \right)} e^{- \beta_R \left( \bar{L}_0 -
		\frac{c}{24} \right)} \right) . 
\end{equation}
where $\beta_L$ and $\beta_R$ are the inverse temperatures of the left and right movers, $L_0$ and $\bar{L}_0$ are the standard Virasoro algebra generators \cite{DiFrancesco:1997nk} and $\mathcal{H}_{\rm CFT}$ is the CFT Hilbert space which is assumed to be the direct sum of state spaces characterized by conformal weights $h$ and $\bar{h}$
\begin{equation}\label{def:Hilbertspace}
	\begin{split}
		\mathcal{H}_{\rm CFT}&=\bigoplus_{h,\bar{h}} V_{h,\bar{h}}, \\
		(L_0,\bar{L}_0)\big{|}_{V_{h,\bar{h}} }&=(h,\bar{h}),\quad {\rm dim}\left(V_{h,\bar{h}}\right)=n_{h,\bar{h}}\in\mathbb{N}.
	\end{split}
\end{equation}
Here $n_{h,\bar{h}}$ counts the degeneracy of the states with conformal weights $h$ and $\bar{h}$. In this work, the full Virasoro symmetry is not essential. The useful property will be that $L_0$ and $\bar{L}_0$ are diagonalizable with non-negative eigenvalues, i.e. $h,\bar{h}\geqslant0$ in \eqref{def:Hilbertspace}.

Using eqs.\,(\ref{def:Z}) and (\ref{def:Hilbertspace}), the torus partition function can be written as a sum of exponential terms
\begin{equation}\label{Z:exp}
	\begin{split}
		Z(\beta_L,\beta_R)=\sum\limits_{h,\bar{h}}n_{h,\bar{h}}\ e^{- 
			\left( h - \frac{c}{24} \right)\beta_L} e^{-  \left( \bar{h} -
			\frac{c}{24} \right)\beta_R}.
	\end{split}
\end{equation}
We assume that (a) the partition function $Z$ is finite when $\beta_L,\beta_R\in(0,\infty)$; (b) $Z$ is invariant under $S$ modular transformation (which we will call ``modular invariance" below)\footnote{In this work we will not use the full modular invariance in the standard sense. So we allow non-integer values of $h-\bar{h}$ in \eqref{Z:exp}.}
\begin{equation}\label{modularinv}
	\begin{split}
		Z(\beta_L,\beta_R)&=Z\left(\frac{4\pi^2}{\beta_L},\frac{4\pi^2}{\beta_R}\right). \\
	\end{split}
\end{equation}
Because of the modular invariance condition \eqref{modularinv}, to know the partition function, it suffices to know it in the regime $\beta_{L}\beta_{R}\geqslant4\pi^2$. For technical reason we will only focus on the open domain $\beta_{L}\beta_R>4\pi^2$, where we will derive some universal inequalities of the partition function.

Now we would like to present the main result of this work.  

\begin{figure}[!ht]
	\centering
	\includegraphics[scale=0.23]{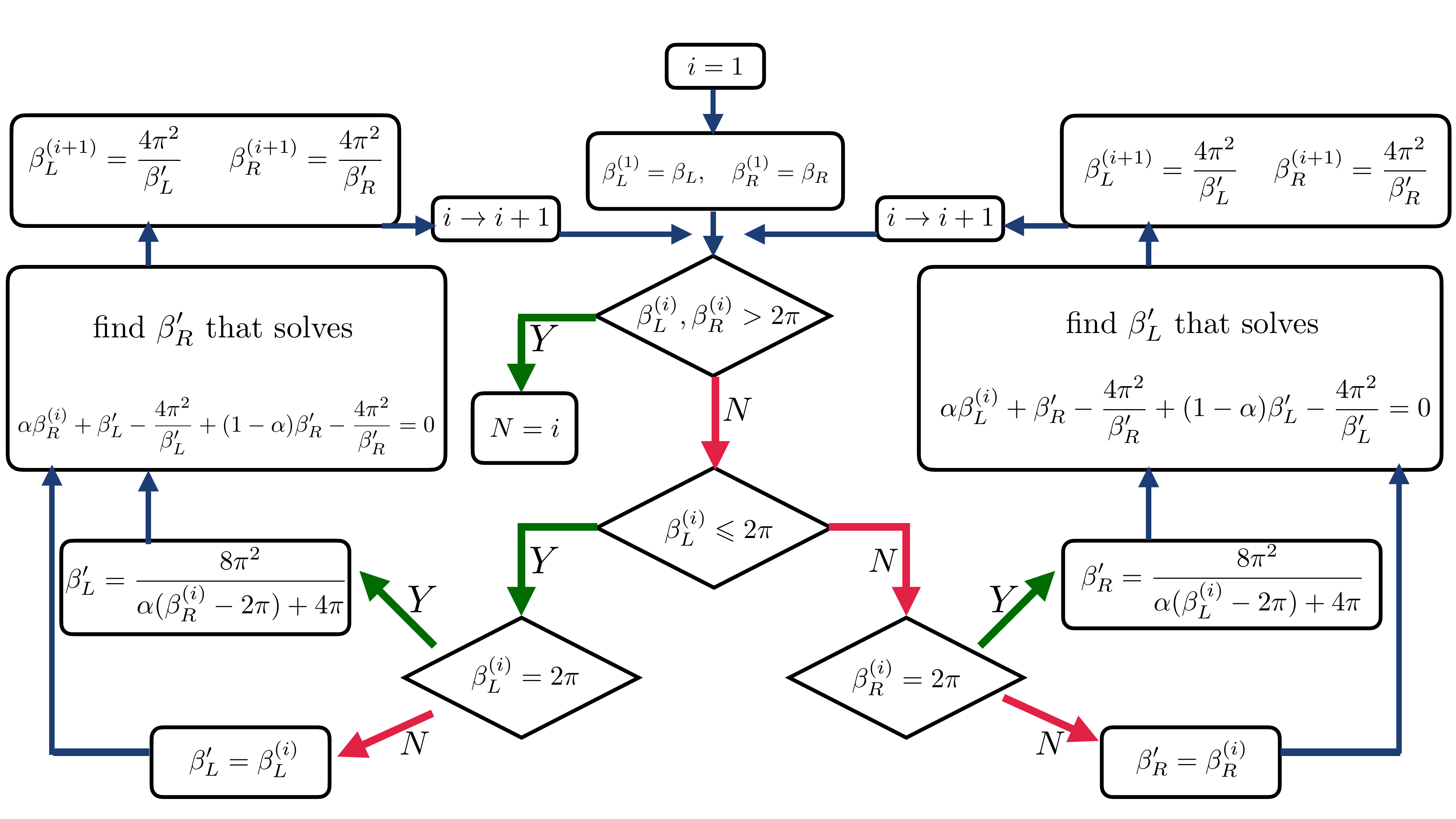}
	\caption{The algorithm for generating $(\beta_{L}^{(i)},\beta_{R}^{(i)})$ ($i=1,2,\ldots,N$). For any $(\beta_{L},\beta_{R})\in\mathcal{D}_\alpha$, this iteration takes finite steps.}
	\label{fig:algorithmbeta}
\end{figure}

\begin{theorem}\label{theorem:logZ}
	Let $\alpha\in(0,1]$ be fixed.  Given any $c\geqslant0$ unitary, modular invariant 2D CFT,  let us define
\begin{equation}\label{def:ZtildeLH}
	\begin{split}
		\tilde{Z}_L(\alpha;\beta_{L},\beta_{R})&:=\sum\limits_{{\rm min}(h,\bar{h})\leqslant\frac{\alpha c}{24}}n_{h,\bar{h}}\ e^{- h\beta_L-\bar{h}\beta_{R}}, \\
		\tilde{Z}_H(\alpha;\beta_{L},\beta_{R})&:=\sum\limits_{h,\bar{h}>\frac{\alpha c}{24}}n_{h,\bar{h}}\ e^{- h\beta_L-\bar{h}\beta_{R}}, \\
	\end{split}
\end{equation}
where the subsripts ``L" and ``H" represent for low twist and high twist. Then, the error term $\mathcal{E}(\beta_{L},\beta_{R})$ of the free energy $\log Z(\beta_{L},\beta_{R})$ (given in \eqref{FE:lowT}) is bounded from above by
	\begin{equation}\label{eq:errorbound}
		\begin{split}
			\mathcal{E}(\beta_{L},\beta_{R})&\leqslant\log\left[\sum_{i=1}^{N}\tilde{Z}_L\left(\alpha;\beta_{L}^{(i)},\beta_{R}^{(i)}\right)+\tilde{Z}_H\left(\alpha;\beta_{L}^{(N)},\beta_{R}^{(N)}\right)\right], \quad(\beta_{L},\beta_R)\in\mathcal{D}_\alpha,\\
		\end{split}
	\end{equation}
	where the domain $\mathcal{D}_\alpha$ is defined as
	\begin{equation}\label{def:Dalpha}
	\begin{split}
		\mathcal{D}_\alpha&:=\left\{(\beta_{L},\beta_{R})\ \Big{|}\ \beta_{R}>2\pi,\ \beta_{L}>\phi_\alpha(\beta_{R});\ {\rm or}\ \beta_{L}>2\pi,\ \beta_{R}>\phi_\alpha(\beta_{L})\right\}, \\
		\phi_\alpha(\beta)&:=\max\left\{\frac{4\pi^2}{\beta},\ \frac{1}{2} \left(\alpha(4\pi - \beta) + \sqrt{\alpha^2(4\pi - \beta)^2 + 16\pi^2 (1 - \alpha)}\right)\right\}\,.
	\end{split}
\end{equation}
	The numbers $N$, $\beta_L^{(i)}$'s and $\beta_R^{(i)}$'s are defined by the algorithm in figure \ref{fig:algorithmbeta}. $N$ is finite for any $(\beta_{L},\beta_{R})\in \mathcal{D}_\alpha$.  Note that $\beta_{L}^{(i)}\,\beta_{R}^{(i)}>4\pi^2\ (i=1,2,\ldots,N-1)$ and $\beta_{L}^{(N)}, \beta_{R}^{(N)}>2\pi\,. $\\
\end{theorem}
\begin{remark}\label{remark:theorem}
	\item (a) The second term in inequality \eqref{eq:errorbound} has a bound in terms of low-energy states \cite{Hartman:2014oaa}:
	\begin{equation}\label{eq:ZHbound}
		\tilde{Z}_H\left(\alpha;\beta_{L},\beta_{R}\right) \leqslant \frac{\widehat{Z}_{L}(\beta)}{1 - e^{-\left(\beta - \frac{4\pi^2}{\beta}\right)\epsilon}}, \quad \beta \equiv \min\left\{\beta_{L}, \beta_{R}\right\} > 2\pi,
	\end{equation}
	where $\widehat{Z}_{L}(\beta)$ is defined in \eqref{def:ZwidehatL}. Eqs.\,\eqref{eq:errorbound} and \eqref{eq:ZHbound} imply the bound \eqref{modifiedTh} stated in the introduction section.

	\item (b) When $\beta_{L},\beta_{R}>2\pi$, according to the algorithm in figure \ref{fig:algorithmbeta}, we simply have $N=1$, $\beta_{L}^{(1)}=\beta_{L}$ and $\beta_{R}^{(1)}=\beta_{R}$ in the right-hand side of \eqref{eq:errorbound}, then the inequality becomes
	\begin{equation}
		\begin{split}
			\log Z(\beta_{L},\beta_{R})-\frac{c}{24}(\beta_{L}+\beta_{R})\leqslant\log Z(\beta_{L},\beta_{R})-\frac{c}{24}(\beta_{L}+\beta_{R}),\quad\beta_{L},\beta_{R}>2\pi,
		\end{split}
	\end{equation}
	which trivially holds. Therefore, the nontrivial part of theorem \ref{theorem:logZ} is the inequality \eqref{eq:errorbound} in the regime where one of $\beta_{L}$ and $\beta_{R}$ is less than $2\pi$.
	
	\item (c) When the theory has a normalizable vacuum, $\mathcal{E}$ has a lower bound given by $\mathcal{E}\geqslant0$, then \eqref{eq:errorbound} becomes the bound on $\abs{\mathcal{E}}$. However, it may happen that the theory does not have a normalizable vacuum, then the error term $\mathcal{E}$ may be negative and large. An example of having a large negative $\mathcal{E}$ is the Liouville CFT with very large central charge.\footnote{The partition function of the Liouville CFT does not depend on the central charge $c$. So if we fix $(\beta_{L},\beta_{R})$ and take very large $c$, the error $\mathcal{E}$ will become negative.}
	
	\item (d) One can verify that $\mathcal{D}_{\alpha_1}\subsetneq\mathcal{D}_{\alpha_2}$ for any $0<\alpha_1<\alpha_2\leqslant1$. When $\alpha=1$, $\mathcal{D}_\alpha$ becomes the domain $\left\{\beta_{L}\beta_{R}>4\pi^2\right\}$.
	
	\item (e) In \eqref{eq:errorbound}, we do not have explicit form for $N$, $\beta_{L}^{(i)}$'s and $\beta_{R}^{(i)}$'s in terms of $\beta_{L}$, $\beta_{R}$ and $\alpha$. The technical reason is that our construction of them is recursive as shown in figure \ref{fig:algorithmbeta}. 
	
	\item (f) It is a non-trivial fact that $\mathcal{D_\alpha}$ is the maximal domain for which the iteration terminates after a finite number of steps, i.e., $N$ is finite.
	
	\item (g) $\alpha$ is a tunable parameter in theorem \ref{theorem:logZ}. In principle, one can choose whatever $\alpha\in(0,1]$ and the theorem always holds. In practice, one may take proper $\alpha$ to minimize the error term $\mathcal{E}$.

	\item (h) In the large central charge limit, the theorem \ref{theorem:logZ} has holographic implications, which will be discussed in sections \ref{subsec:HKS} and \ref{subsec:HKSM}.
	
	\item (i) In this paper, the arguments and results remain unchanged under a slight modification of the definitions of $\tilde{Z}_L$ and $\tilde{Z}_H$, as specified below:
	\begin{equation}
		\begin{split}
			\tilde{Z}_L&:=\sum\limits_{{\rm min}(h,\bar{h})\leqslant\frac{\alpha c}{24}}(\ldots) \quad\longrightarrow\quad \tilde{Z}_L:=\sum\limits_{{\rm min}(h,\bar{h})<\frac{\alpha c}{24}}(\ldots), \\
			\tilde{Z}_H&:=\sum\limits_{h,\bar{h}>\frac{\alpha c}{24}}(\ldots) \quad\longrightarrow\quad \tilde{Z}_H:=\sum\limits_{h,\bar{h}\geqslant\frac{\alpha c}{24}}(\ldots).
		\end{split}
	\end{equation}
\end{remark}

The proof of theorem \ref{theorem:logZ} relies on a technical lemma which we will present in section \ref{section:lemma}. Once the lemma is proven, theorem \ref{theorem:logZ} will follow as an immediate consequence. See the proof below lemma \ref{lemma:functionbound} for details.

\subsection{An inequality for modular invariant functions}\label{section:lemma}
In this subsection we present an inequality that holds for a class of modular invariant functions. The unitary CFT partition functions, after factorizing out the $e^{\frac{c}{24}(\beta_{L}+\beta_{R})}$ factor, belong to this class of functions under a proper change of variables, which will lead to the proof of theorem \ref{theorem:logZ}. However, there may be functions that belong to this function class, but not necessarily satisfy all the properties that a unitary CFT partition function should satisfy. So our argument below may have potential applications to some non-unitary CFTs.

For CFT partition functions, the $S$-modular transformation is $\beta_{L,R}\rightarrow4\pi^2/\beta_{L,R}$ and the self-dual line is $\beta_{L}\beta_{R}=4\pi^2$. To simplify the notation, we rescale the variables by $(2\pi)^{-1}$: $x=\beta_{L}/2\pi$, $y=\beta_{R}/2\pi$. Then the $S$-modular transformation becomes $x,y\rightarrow 1/x,1/y$ and the self-dual line becomes $xy=1$.

Let $\alpha\in(0,1]$ be fixed. Under the above notation, the domain $\mathcal{D}_\alpha$ defined in \eqref{def:Dalpha} is rescaled by
\begin{equation}\label{def:Omegaalpha}
	\begin{split}
		\Omega_\alpha &= \left\{(x,y) \,\Big{|}\, x > 1,\ y > f_\alpha(x);\ {\rm or}\ y > 1,\ x > f_\alpha(y) \right\}, \\
		f_\alpha(x)&=\max\left\{\frac{1}{x},\ \frac{1}{2} \left(\alpha(2 - x) + \sqrt{\alpha^2(2 - x)^2 + 4 (1 - \alpha)}\right)\right\}.
	\end{split}
\end{equation}
Recall that for each $(\beta_{L},\beta_{R})\in\mathcal{D}_\alpha$, we generated a collection of points $(\beta_{L}^{(i)},\beta_{R}^{(i)})$ using the algorithm shown in figure \ref{fig:algorithmbeta}. Similarly, after rescaling by $(2\pi)^{-1}$, for each $(x,y)\in\Omega_\alpha$, we have a collection of points $(x_i,y_i)$ generated by the algorithm shown in figure \ref{fig:algorithmxy}.
\begin{figure}[!ht]
	\centering
	\includegraphics[scale=0.23]{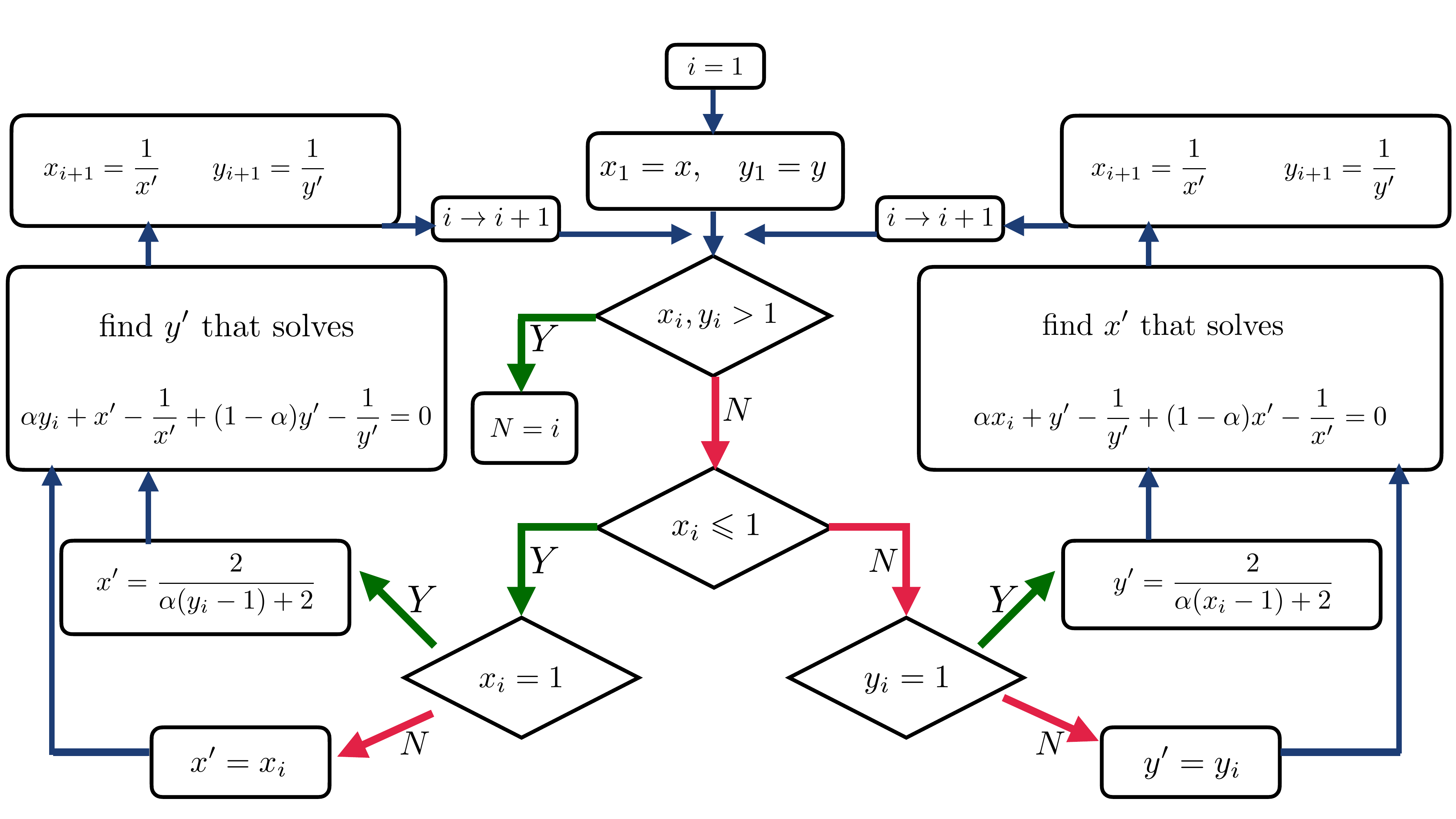}
	\caption{The algorithm for generating $(x_i,y_i)$ ($i=1,2,\ldots,N$). For any $(x,y)\in\Omega_\alpha$, this iteration takes finite steps.}
	\label{fig:algorithmxy}
\end{figure}

Now we introduce a class of positive functions $g(x,y)$ defined on $(0,\infty)^2$ satisfying the following properties
\begin{itemize}
	\item (Modular invariance) There exists $\gamma\geqslant0$ such that
	\begin{equation}\label{g:modinv}
		\begin{split}
			e^{\gamma(x+y)}\,g\left(x,y\right)=e^{\gamma\left(\frac{1}{x}+\frac{1}{y}\right)}\,g\left(\frac{1}{x},\frac{1}{y}\right);
		\end{split}
	\end{equation}
	\item (Twist-gap property) $g$ can be split into two parts
	\begin{equation}\label{positivity}
		\begin{split}
			g(x,y)=g_0(x,y)+g_1(x,y),
		\end{split}
	\end{equation}
	where both $g_0$ and $g_1$ are non-negative. In addition, $g_1$ satisfies the following inequality
	\begin{equation}\label{twistprop}
		\begin{split}
			g_1(x,y)\leqslant e^{-\alpha\gamma\left(x-x'\right)}g_1(x',y),\quad\forall0<x'\leqslant x, \\
			g_1(x,y)\leqslant e^{-\alpha\gamma\left(y-y'\right)}g_1(x,y'),\quad\forall 0<y'\leqslant y; \\
		\end{split}
	\end{equation}
\end{itemize}
Then the following inequality holds for $g(x,y)$.
\begin{lemma}\label{lemma:functionbound}	
	For $g(x,y)$ satisfying the above properties, the following inequality holds:
	\begin{equation}\label{inequality:general}
		\begin{split}
			g(x,y)&\leqslant\sum_{i=1}^{N}g_0\left(x_i,y_i\right)+g_1\left(x_{N},y_{N}\right),\quad\forall(x,y)\in\Omega_\alpha,
		\end{split}
	\end{equation}
	where the domain $\Omega_\alpha$ is defined by \eqref{def:Omegaalpha}. In the inequality \eqref{inequality:general}, $N$, $x_i$'s and $y_i$'s are defined by the algorithm in figure \ref{fig:algorithmxy}. $N$ is finite for any $(x,y)\in\Omega_\alpha$.
\end{lemma}

The proof of lemma \ref{lemma:functionbound} is technical. We postpone it to section \ref{sec:prooflemma}. Here let us show how lemma \ref{lemma:functionbound} implies theorem \ref{theorem:logZ}.
\newline
\newline\textbf{Proof of theorem \ref{theorem:logZ}:}
\newline For any $c\geqslant0$ unitary CFT partition function $Z(\beta_{L},\beta_{R})$:
\begin{equation}
	\begin{split}
		Z(\beta_{L},\beta_{R})=\sum\limits_{h,\bar{h}}n_{h,\bar{h}}\,e^{-\left(h-\frac{c}{24}\right)\beta_{L}-\left(\bar{h}-\frac{c}{24}\right)\beta_{R}},
	\end{split}
\end{equation}
we have the following identification
\begin{equation}\label{identification}
	\begin{split}
		&\gamma=\frac{\pi c}{12},\quad x=\frac{\beta_{L}}{2\pi},\quad y=\frac{\beta_{R}}{2\pi}, \\
		&g(x,y)=e^{-\frac{c}{24}(\beta_{L}+\beta_{R})}Z(\beta_{L},\beta_{R}), \\
		&g_0(x,y)=\sum\limits_{\text{min}(h,\bar{h})\leqslant\frac{\alpha c}{24}}n_{h,\bar{h}}\,e^{-h\beta_{L}-\bar{h}\beta_{R}}\equiv\tilde{Z}_L(\alpha;\beta_{L},\beta_{R}),  \\
		&g_1(x,y)=\sum\limits_{h,\bar{h}>\frac{\alpha c}{24}}n_{h,\bar{h}}\,e^{-h\beta_{L}-\bar{h}\beta_{R}}\equiv\tilde{Z}_H(\alpha;\beta_{L},\beta_{R}).
	\end{split}
\end{equation}
One can explicitly check that all the conditions of lemma \ref{lemma:functionbound} are satisfied. Then plug \eqref{identification} into \eqref{inequality:general} and take logarithm on both sides the inequality, we get
\begin{equation}
	\begin{split}
		\mathcal{E}(\beta_{L},\beta_{R})\leqslant \log\left[\sum_{i=1}^{N}\tilde{Z}_L\left(\alpha;2\pi x_i,2\pi y_i\right)+\tilde{Z}_H\left(\alpha;2\pi x_{N},2\pi y_{N}\right)\right].
	\end{split}
\end{equation}
Then by identification
\begin{equation}
	\begin{split}
		\beta_{L}^{(i)}=2\pi x_i,\quad\beta_{R}^{(i)}=2\pi y_i,\quad i=1,2,\ldots,N,
	\end{split}
\end{equation}
we get \eqref{eq:errorbound}. Therefore, theorem \ref{theorem:logZ} follows from lemma \ref{lemma:functionbound}.

\subsection{Application to 3d gravity: proof of HKS conjecture}\label{subsec:HKS}
The AdS/CFT correspondence suggests the existence of a 2D CFT dual to Einstein gravity in AdS$_3$, assuming a quantum theory of 3D gravity exists. However, finding a suitable CFT partition function for 3D gravity remains an open problem. To gain an intuitive understanding of the expected form of the CFT partition function for 3D gravity, one may consider the semiclassical regime, where the Newton coupling is very small compared to the AdS length scale:
\begin{equation}
	\begin{split}
		G_{\rm N} \ll \ell_{\text{AdS}}.
	\end{split}
\end{equation}
In this regime, it is expected that the saddle-point approximation holds for the partition function:
\begin{equation}
	\begin{split}
		Z(\beta_{L},\beta_{R}) \approx \sum\limits_{g^{(cl)}_{\mu\nu}} e^{-S_{\rm gravity}[g_{\mu\nu}^{(cl)},\beta_{L},\beta_{R}]},
	\end{split}
\end{equation}
where \(S_{\rm gravity}\) is the Einstein-Hilbert action with a negative cosmological constant, and \(g^{(cl)}_{\mu\nu}\) represents classical solutions to the Einstein equation. The temperatures \(\beta_{L}\) and \(\beta_{R}\) appear as the asymptotic boundary conditions for the metric at infinity.

For a given \((\beta_{L},\beta_{R})\), there exists an SL$(2;\mathbb{Z})$ family of Euclidean solutions, leading to \cite{Maldacena:1998bw}:
\begin{equation}\label{BHsolutions}
	\begin{split}
		&S_{\rm gravity}=\frac{i \pi c}{12}\left(\gamma\cdot\tau-\gamma\cdot\bar{\tau}\right), \\
		&\tau=\frac{i\beta_{L}}{2\pi},\ \bar{\tau}=-\frac{i\beta_{R}}{2\pi},\quad \gamma \in {\rm SL}(2;\mathbb{Z}),
	\end{split}
\end{equation}
where \(c = \frac{3\ell_{\text{AdS}}}{2G_{\rm N}} \gg 1\) \cite{Brown:1986nw}, and $\gamma$ is the SL$(2;\mathbb{Z})$ transformation defined by \(\gamma\cdot\tau \equiv \frac{a\tau+b}{c\tau+d}\) with integer parameters satisfying \(ad - bc = 1\) (the parameter \(c\) here is the integer parameter instead of the central charge). When \(\beta_{L}\) and \(\beta_{R}\) are real and positive, the minimum of \(\text{Re}(S_{\rm gravity})\) corresponds to either thermal AdS$_3$ or the BTZ black hole:
\begin{equation}
	\begin{split}
		\min\limits_{\gamma \in {\rm SL}(2;\mathbb{Z})} \text{Re}(S_{\rm gravity}) = \min \left\{-\frac{c}{24}(\beta_{L} + \beta_{R}),\ -\frac{\pi^2c}{6}\left(\frac{1}{\beta_{L}} + \frac{1}{\beta_{R}}\right)\right\}.
	\end{split}
\end{equation}
It turns out that \(S_{\rm gravity}\) becomes real when \(\text{Re}(S_{\rm gravity})\) takes its minimum. Therefore, the saddle-point analysis predicts two phases at large \(c\), for which the free energy is given by:
\begin{equation}\label{gravity:logZ}
	\begin{split}
		\log Z(\beta_{L},\beta_{R}) = \begin{cases}
			\frac{c}{24}(\beta_{L} + \beta_{R}) + \text{o}(c), & \beta_{L} \beta_{R} > 4\pi^2, \\
			\\
			\frac{\pi^2c}{6} \left(\frac{1}{\beta_{L}} + \frac{1}{\beta_{R}}\right) + \text{o}(c), & \beta_{L} \beta_{R} < 4\pi^2,
		\end{cases}
	\end{split}
\end{equation}
where \(\text{o}(c)\) denotes subleading terms in large \(c\) with \(\beta_{L}\) and \(\beta_{R}\) fixed. The Hawking-Page phase transition \cite{Hawking:1982dh} happens on the self-dual line $\beta_{L}\beta_R=4\pi^2$.

We observe that in \eqref{gravity:logZ}, the leading term of the free energy is Cardy-like: it corresponds to the logarithm of the vacuum term in the direct (low-temperature, \(\beta_{L} \beta_{R}> 4\pi^2\)) or dual (high-temperature, \(\beta_{L} \beta_{R} < 4\pi^2\)) expansion of the 2D CFT partition function. However, this result is not directly implied by the standard Cardy formula \cite{Cardy:1986ie}, where the low-temperature limit with fixed \(c\) is considered. This raises a natural question: what type of CFTs produce such free energy?

The answer to this question was proposed by Hartman, Keller and Stoica in \cite{Hartman:2014oaa}. They considered a unitary, modular invariant CFT satisfying the following properties
\begin{itemize}
	\item The theory has a normalizable vacuum.
	\item The central charge of the theory is tunable and allows the limit \(c \rightarrow \infty\).
	\item There exists a fixed \(\epsilon > 0\), and the spectrum of scaling dimensions \(\Delta := h + \bar{h}\) below \(\frac{c}{12} + \epsilon\) is sparse, such that:
	\begin{equation}\label{HKSenergy}
		\sum\limits_{h + \bar{h} \leqslant \frac{c}{12} + \epsilon} n_{h,\bar{h}} \ e^{-(h + \bar{h})\beta} \leqslant A(\beta), \quad \beta > 2\pi,
	\end{equation}
	where \(A(\beta) < \infty\) for \(\beta > 2\pi\), and \(A\) does not depend on \(c\).
	
	\item The spectrum of twist (defined by \(\min\left\{2h, 2\bar{h}\right\}\)) below \(\frac{c}{12}\) is sparse, such that:
	\begin{equation}\label{HKStwist}
		\sum\limits_{h \, {\rm or} \, \bar{h} \leqslant \frac{c}{24}} n_{h,\bar{h}} \ e^{-h\beta_{L} - \bar{h}\beta_{R}} \leqslant B(\beta_{L}, \beta_{R}), \quad \beta_{L} \beta_{R} > 4\pi^2,
	\end{equation}
	where \(B(\beta_{L}, \beta_{R}) < \infty\) for \(\beta_{L} \beta_{R} > 4\pi^2\), and \(B\) does not depend on \(c\).
\end{itemize}
The first three conditions imply \eqref{gravity:logZ} when there is no angular potential, i.e., \(\beta_{L} = \beta_{R} = \beta\). To extend \eqref{gravity:logZ} to mixed temperatures \((\beta_{L}, \beta_{R})\), HKS imposed a sparseness condition on the spectrum of twist (the last condition) and conjectured that \eqref{gravity:logZ} should hold away from the self-dual line \(\beta_{L} \beta_{R} = 4\pi^2\):

\begin{conjecture}\label{conjecture:HKS}
	(HKS \cite{Hartman:2014oaa}) Under the above assumptions, the free energy satisfies the following asymptotic behavior in the limit \(c \rightarrow \infty\) with \(\beta_{L}\) and \(\beta_{R}\) fixed:
	\begin{equation}\label{eq:HKSconjecture}
		\begin{split}
			\log Z(\beta_{L}, \beta_{R}) = \begin{cases}
				\frac{c}{24}(\beta_{L} + \beta_R) + {\rm O}(1), & \beta_{L} \beta_R > 4\pi^2, \\
				\frac{\pi^2c}{6}\left(\frac{1}{\beta_{L}} + \frac{1}{\beta_{R}}\right) + {\rm O}(1), & \beta_{L} \beta_R < 4\pi^2. \\
			\end{cases}
		\end{split}
	\end{equation}
	Here, by {\rm O(1)} we mean that the error term \(\mathcal{E}(\beta_{L}, \beta_{R})\) is of order one in $c$, i.e. it is bounded in absolute value by some finite number that only depends on $\beta_{L}$ and $\beta_{R}$.\footnote{We can relax the conditions \eqref{HKSenergy} and \eqref{HKStwist} by allowing the functions $A(\beta)$ and $B(\beta_{L}, \beta_{R})$ to depend on $c$, provided that this dependence remains $\text{o}(c)$ as $c \rightarrow \infty$. In this case, \eqref{eq:HKSconjecture} in conjecture \ref{conjecture:HKS} should be replaced by \eqref{gravity:logZ}.}
\end{conjecture}
In \cite{Hartman:2014oaa}, HKS presented an iterative argument aimed at proving conjecture \ref{conjecture:HKS}. Their approach involves estimating both the partition function and the entropy. In each iteration step, an improved bound on the entropy is derived from the existing bound on the partition function, and then an improved bound on the partition function is obtained from the new bound on the entropy. After three iterations, HKS demonstrated that the error term in \eqref{eq:HKSconjecture} is \( \text{O}(1) \) in \(c\) for \((\beta_{L}, \beta_{R})\) in a slightly smaller domain than the one proposed in \eqref{eq:HKSconjecture} (see figure 2 in \cite{Hartman:2014oaa}). They conjectured that infinite iterations of their argument would lead to \eqref{eq:HKSconjecture}, with the error term \( \text{O}(1) \) in \(c\) for any \(\beta_{L}\beta_{R} \neq 4\pi^2\).
	
However, already by the third iteration, their estimate became quite complicated, and the regime of \((\beta_{L}, \beta_{R})\) had to be computed numerically. It remains unclear how to continue the iteration process. Nevertheless, the proof of lemma \ref{lemma:functionbound}, which we present in this work, is inspired by \cite{Hartman:2014oaa} and is also iterative in nature. The main difference between our approach and theirs is that we only estimate the partition function, without attempting to estimate the entropy. Our approach enables us to carry out an infinite number of iterations, ultimately leading to the proof of the HKS conjecture. For details, refer to section \ref{sec:prooflemma}, and see figure \ref{fig:5} for a visual representation of the iteration process.
\begin{figure}[!ht]
	\centering
	\includegraphics[scale=0.16]{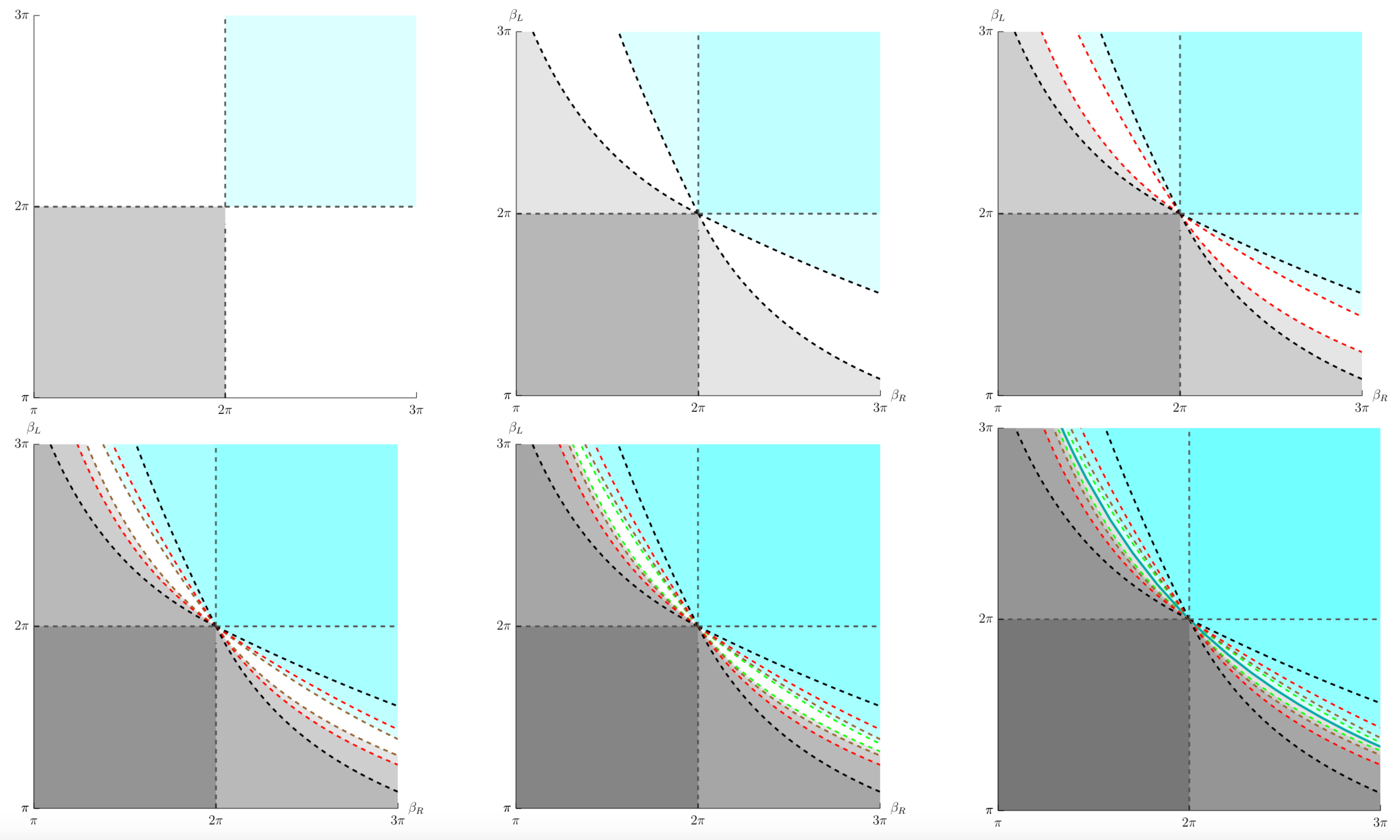}
		\caption{We show how the region of universality grows after each iterations step.  In all the figures,  the cyan region is dominated by thermal AdS while the gray shaded region is dominated by black hole saddle.  The left most figure depicts $0$th iteration, where the universal regime is  union of $\beta_L,\beta_R>2\pi$ and $\beta_L,\beta_R<2\pi$.  With increasing iterations,  more regions become universal.  We denote the boundary of this universal region with dotted lines.  The three figures on the top panel have already obtained in HKS\cite{Hartman:2014oaa}.  In the bottom panel,  we add results from more iteration steps including the one,  obtained analytically after infinite iteration. After infinite iterations,  we cover the whole plane except $\beta_L\beta_R=4\pi^2$, denoted by the solid green line in the rightmost figure on the bottom panel.  }
	\label{fig:5}
\end{figure}

Let us now demonstrate how conjecture \ref{conjecture:HKS} follows from theorem \ref{theorem:logZ} as an immediate consequence. It suffices to only consider the case \(\beta_{L}\beta_{R} > 4\pi^2\), since the case \(\beta_{L}\beta_{R} < 4\pi^2\) follows from modular invariance \eqref{modularinv}.

According to remark \ref{remark:theorem}-(c), we have \(\mathcal{E}(\beta_{L},\beta_{R}) \geqslant 0\) by the first assumption here. Thus, the key point of conjecture \ref{conjecture:HKS} is that \(\mathcal{E}(\beta_{L},\beta_{R})\) is bounded from above by a \(c\)-independent number for \(\beta_{L}\beta_{R} > 4\pi^2\). According to remark \ref{remark:theorem}-(d), this corresponds to the case \(\alpha = 1\) in theorem \ref{theorem:logZ}. Then, from \eqref{eq:errorbound}, it suffices to verify that:
\begin{equation}\label{HKSconj:tocheck}
	\begin{split}
		\tilde{Z}_L(1;\beta_{L},\beta_{R}) &\leqslant B(\beta_{L},\beta_{R}) < \infty, \quad \beta_{L}\beta_{R} > 4\pi^2,\\
		\tilde{Z}_H(1;\beta_{L},\beta_{R}) &\leqslant P(\beta_{L},\beta_{R}) < \infty, \quad \beta_{L}, \beta_{R} > 2\pi,
	\end{split}
\end{equation}
where the quantities \(\tilde{Z}_L\) and \(\tilde{Z}_H\) are defined in \eqref{def:ZtildeLH}, and the functions \(B(\beta_{L},\beta_{R})\) and \(P(\beta_{L},\beta_{R})\) should be independent of \(c\).

The first inequality of \eqref{HKSconj:tocheck} is simply a restatement of \eqref{HKStwist}, which holds by assumption.

The second inequality of \eqref{HKSconj:tocheck} can be derived in two steps. In the first step, we derive the bound for \(\beta_{L} = \beta_{R} = \beta\):
\begin{equation}\label{checkcondition:step1}
	\begin{split}
		Z(\beta,\beta) e^{-\frac{c}{12} \beta} \leqslant \frac{1}{1 - e^{-\epsilon\left(\beta - \frac{4\pi^2}{\beta}\right)}} A(\beta), \quad \beta > 2\pi.
	\end{split}
\end{equation} 
This follows from \eqref{HKSenergy} and the original argument by HKS.\footnote{For the detailed derivation, see the argument from eq.\,(2.5) to eq.\,(2.10) in \cite{Hartman:2014oaa}.}

In the second step, we apply the unitarity condition. Consider the regime \(\beta_{L} \geqslant \beta_{R} > 2\pi\). We obtain the following bound for \(\tilde{Z}_H(1; \beta_{L}, \beta_{R})\):
\begin{equation}\label{checkcondition:step2}
	\begin{split}
		\tilde{Z}_H(1; \beta_{L}, \beta_{R}) \leqslant \tilde{Z}_H(1; \beta_{R}, \beta_{R}) \leqslant Z(\beta_{R}, \beta_{R}) e^{-\frac{c}{12} \beta_{R}} \leqslant \frac{1}{1 - e^{-\epsilon\left(\beta_{R} - \frac{4\pi^2}{\beta_{R}}\right)}} A(\beta_{R}).
	\end{split}
\end{equation}
In the first inequality, we use the fact that all terms are positive and that \(e^{-h \beta_{L}} \leqslant e^{-h \beta_{R}}\) (for \(h \geqslant 0\)) as a consequence of unitarity. In the second inequality, we bound \(\tilde{Z}_H\) by the sum over all CFT states. The final inequality comes from \eqref{checkcondition:step1}. A similar argument applies for the regime \(\beta_{R} \geqslant \beta_{L} > 2\pi\). 

Thus, we arrive at the second inequality of \eqref{HKSconj:tocheck}, with
\begin{equation}\label{HKS:P}
	\begin{split}
		P(\beta_{L}, \beta_{R}) = \frac{1}{1 - e^{-\epsilon\left(\beta - \frac{4\pi^2}{\beta}\right)}} A(\beta), \quad \beta \equiv \min\left\{\beta_{L}, \beta_{R}\right\} > 2\pi.
	\end{split}
\end{equation}
Now, it follows from theorem \ref{theorem:logZ} and \eqref{HKSconj:tocheck} that the right-hand side of \eqref{eq:errorbound} is independent of \(c\). Consequently, when we take the limit \(c \rightarrow \infty\) with \(\beta_{L}\) and \(\beta_{R}\) fixed (\(\beta_{L} \beta_{R} > 4\pi^2\)), the free energy is dominated by the leading term \(\frac{c}{24}(\beta_{L} + \beta_{R})\). This completes the proof of the HKS conjecture.

Before the end of this subsection, we would like to mention that in the work by Anous, Mahajan, and Shaghoulian \cite{Anous:2018hjh}, a slightly stronger twist sparseness condition was imposed:
\begin{itemize}
	\item There exists \(\epsilon > 0\), such that the spectrum of twist below \(\frac{c}{12} + \epsilon\) is sparse.
\end{itemize}
Due to the presence of a positive \(\epsilon\), the HKS argument (for the case without angular potential) extends to the mixed-temperature regime. However, this argument breaks down immediately when \(\epsilon = 0\), which corresponds to the original HKS conjecture \cite{Hartman:2014oaa}.

It is unclear whether the presence or absence of such a positive \(\epsilon\) is physically interesting. Nevertheless, the method we developed to handle the case of \(\epsilon = 0\) proves to be useful for a more general class of CFTs, where the spectrum of twist is sparse below \(\frac{\alpha c}{12}\), with \(\alpha \in (0,1]\) fixed. The case \(\alpha = 1\) reduces to the sparseness condition of HKS. For \(\alpha < 1\), theorem \ref{theorem:logZ} still enables us to explore the mixed-temperature phase diagram at large \(c\). This will be discussed in the next subsection.

\subsection{Large $c$ phase diagram with $\alpha\leqslant1$}\label{subsec:HKSM}
In this subsection, we would like to consider a larger class of unitary, modular invariant CFTs compare to the one in the previous section. We keep the all the assumptions of HKS, except the last one, the sparseness condition on the spectrum of twist, which is relaxed to the following assumption:
\begin{itemize}
	\item There exists a fixed $\alpha\in(0,1]$, such that the spectrum of twist ($\min\left\{2h,2\bar{h}\right\}$) below $\frac{\alpha c}{12}$ is sparse in the following sense
	\begin{equation}\label{alphatwist}
		\begin{split}
			\sum\limits_{\text{min}(h,\bar{h})\leqslant\frac{\alpha c}{24}}n_{h,\bar{h}}\ e^{- h\beta_{L}-\bar{h}\beta_R}\leqslant B(\beta_{L},\beta_{R}),\quad\beta_L\beta_{R}>4\pi^2, \\
		\end{split}
	\end{equation}
	where $B(\beta_L,\beta_{R})<\infty$ for $\beta_L\beta_{R}>4\pi^2$. $B$ does not depend on $c$.
\end{itemize}
The case of $\alpha=1$ is the same as the one that HKS considered in \cite{Hartman:2014oaa}.

Using the same argument as in the proof of HKS conjecture in the previous subsection, we immediately conclude that the large-$c$ behavior of the free energy is universal for all $(\beta_{L},\beta_{R})$ in domain $\mathcal{D}_\alpha$.
\begin{corollary}\label{cor:alpha}
	Suppose the HKS twist sparseness condition replaced by the $\alpha$-dependent one, eq.\,\eqref{alphatwist}, while the other conditions remain the same. Then the CFT free energy satisfies the following asymptotic behavior in the limit $c\rightarrow\infty$ with $\beta_{L}$ and $\beta_{R}$ fixed:
	\begin{equation}\label{eq:logz:alpha}
		\begin{split}
			\log Z(\beta_{L},\beta_{R})=\begin{cases}
				\frac{c}{24}(\beta_{L}+\beta_R)+{\rm O}(1), & (\beta_{L},\beta_{R})\in \mathcal{D}_\alpha, \\
				\frac{\pi^2c}{6}\left(\frac{1}{\beta_{L}}+\frac{1}{\beta_{R}}\right)+{\rm O}(1), & \left(\frac{4\pi^2}{\beta_{L}},\frac{4\pi^2}{\beta_{R}}\right)\in \mathcal{D}_\alpha. \\
			\end{cases}
		\end{split}
	\end{equation} 
	Here, {\rm O(1)} means that the error term \(\mathcal{E}(\beta_{L}, \beta_{R})\) is of order one in $c$, i.e. it is bounded in absolute value by some finite number that only depends on $\beta_{L}$ and $\beta_{R}$.
\end{corollary}
Depending on $\alpha$, there are three possible scenarios
\begin{itemize}
	\item $\alpha = 1$: In this case, $\mathcal{D}_1=\{\beta_{L}\beta_{R}>4\pi^2\}$. Therefore, the maximal-validity-domain of \eqref{eq:logz:alpha}, the union of $\mathcal{D}_1$ and its image of $S$-modular transformation, is the whole first quadrant of the $(\beta_{L},\beta_{R})$ plane except for the self-dual line $\beta_{L}\beta_{R}=4\pi^2$.  See figure \ref{fig:1}-left.
	\item $\frac{1}{2} < \alpha < 1$: In this case, the maximal-validity-domain of \eqref{eq:logz:alpha} is strictly smaller than $\{\beta_{L}\beta_{R}\neq4\pi^2\}$. There is a transition point $\beta_* = \frac{2\pi}{2\alpha - 1}$ for the boundary of $\mathcal{D}_\alpha$ (decribed by the function $\phi_\alpha$ in \eqref{def:Dalpha}). When $\beta_{L}\geqslant\beta_*$ or $\beta_{R}\geqslant\beta_*$, the boundary function $\phi_\alpha(\beta_{L,R})$ equals $4\pi^2/\beta_{L,R}$, so this part of the boundary is still the self-dual line $\beta_{L}\beta_{R}=4\pi^2$. However, when $2\pi<\beta_{L,R} \leqslant \beta_*$, the boundary function $\phi_\alpha(\beta_{L,R})$ starts to take the other value in \eqref{def:Dalpha} and deviates from the self-dual line. The function $\phi_\alpha(\beta)$ is continuous, but its first derivative is discontinuous at $\beta = \beta_*$. See figure \ref{fig:1}-right for the maximal-validity-domain of this case.
	\item $0 < \alpha \leqslant \frac{1}{2}$: In this case, there is no transition for the boundary of the maximal-validity-domain of \eqref{eq:logz:alpha}: the function $\phi_\alpha(\beta)$ always takes the second value in \eqref{def:Dalpha}. See figure \ref{fig:2} for the maximal-validity-domain of this case.
\end{itemize}

\section{Proof of lemma \ref{lemma:functionbound}}\label{sec:prooflemma}
In this section, we aim to prove lemma \ref{lemma:functionbound}. Before starting the proof, we would like to provide a conceptual overview of why lemma \ref{lemma:functionbound} holds.

In inequality \eqref{inequality:general}, we observe that the first term on the right-hand side involves a finite sum. The reason why it includes many points \((x_i, y_i)\) is that the bound on \(g(x, y)\) is established through several steps. 

In the first step, we split \(g(x, y)\) into two parts, interpreted as low-twist and high-twist parts:
$$g(x, y) = g_0(x, y) + g_1(x, y).$$
We leave the first term \(g_0\) unchanged, yielding \(x_1 = x\) and \(x_2 = y\) in \eqref{inequality:general}. Then, we estimate the bound on the second term \(g_1\). Using modular invariance, \(g_1(x_1, y_1)\) is bounded by \(g\) evaluated at another point \((x_2, y_2)\), which moves us closer to the regime \(\{x, y > 1\}\):
$$g_1(x_1, y_1) \leqslant g(x_2, y_2).$$
An important point here is that \((x_2, y_2)\) does not depend on which specific function \(g\) is being considered: it is algorithmically determined by \((x_1, y_1)\).

Next, we estimate \(g(x_2, y_2)\) in the same manner. This, combined with the previous process, gives the bound:
$$g(x, y) \leqslant g_0(x_1, y_1) + g_0(x_2, y_2) + g(x_3, y_3),$$
where \((x_3, y_3)\) is determined by \((x_2, y_2)\) and further approaches the regime \(\{x, y > 1\}\).

By repeating the above process, we obtain a sequence of points 
$$(x_1, y_1),\ (x_2, y_2),\ (x_3, y_3),\ \ldots$$ 
which progressively move closer to \(\{x, y > 1\}\). A nontrivial observation is that after a finite number of steps, we reach a point \((x_N, y_N)\) that lies within the regime \(\{x, y > 1\}\). This special property holds for points \((x, y)\) in the domain \(\Omega_\alpha\), as defined in \eqref{def:Omegaalpha}. Thus, inequality \eqref{inequality:general} follows.

We split the proof into two main parts.

In the first part, we introduce an iteration problem involving a sequence of functions 
$$f_0(x), f_1(x), f_2(x), \ldots \quad (x \geqslant 1)\,,$$ 
and demonstrate that they exhibit monotonicity properties. We also determine the pointwise limit of this sequence. This is discussed in section \ref{subsec:iteration}.

In the second part, we prove lemma \ref{lemma:functionbound} inductively for \((x, y)\) in the domains:
$$\Omega_\alpha^{(0)} \subset \Omega_\alpha^{(1)} \subset \Omega_\alpha^{(2)} \subset \ldots$$
The maximal domain \(\Omega_\alpha\), defined by \eqref{def:Omegaalpha}, is the limit of this sequence:
$$\Omega_\alpha = \bigcup\limits_{k=0}^\infty \Omega_\alpha^{(k)}.$$
In each induction step, the boundary of the domain \(\Omega_\alpha^{(k)}\) is described by the function \(f_k\), which arises in the iteration problem in the first part of the proof. This is addressed in section \ref{section:derivationiteration}.

\subsection{Step 1: an iteration problem}\label{subsec:iteration}
We introduce a sequence of functions $\{f_k(x)\}$ for $x \geqslant 1$ by the following recursion relation:
\begin{equation}\label{f:recursion0}
	\begin{split}
		f_0(x) &\equiv 1, \\
		f_{k+1}(x) &= \max \left\{ y \in (0,1] \,\Bigg{|}\, F_k(y) = \alpha x \right\}, \quad k \geqslant 0,
	\end{split}
\end{equation}
where the function $F_k(y)$ is defined by:
\begin{equation}\label{def:Fk}
	F_k(y) \equiv \frac{1}{y} - y - \frac{1-\alpha}{f_{k}\left(\frac{1}{y}\right)} + f_{k}\left(\frac{1}{y}\right) \quad (0 < y \leqslant 1).
\end{equation}
It is not obvious that the above sequence is well-defined, but we will show below that it is.

The motivation for introducing the above iteration problem was mentioned at the beginning of section \ref{sec:prooflemma}: lemma \ref{lemma:functionbound} will be proven inductively domain by domain; in each step, the boundary of domain $\Omega_\alpha^{(k)}$ is described by the function $f_k$.

We would like to prove the following properties for $\{f_k\}$:
\begin{proposition}\label{prop:fsequence}
	Let $0 < \alpha \leqslant 1$ be fixed. The function sequence $\{f_k\}$ is well-defined, and it satisfies the following properties:
	
	(a) For any $k \geqslant 1$, $f_k(x)$ is a strictly decreasing smooth function on $[1, \infty)$. We have $f_k(1) = 1$ and $f_k(\infty) = 0$. For all $x > 1$, $f_k(x) > \frac{1}{x}$.
	
	(b) For any $k \geqslant 0$, $F_k(y)$ (defined in \eqref{def:Fk}) is a strictly decreasing smooth function on $(0, 1]$. We have $F_k(0) = \infty$ and $F_k(1) = \alpha$.
	
	(c) For all $k \geqslant 0$ and $x > 1$, we have $f_{k+1}(x) < f_k(x)$.
	
	(d) The sequence $\{f_k\}$ has a pointwise limit, given by
	\begin{equation}\label{f:infinity}
		f_\infty(x) = \max\left\{f_\alpha^{(1)}(x),\ f_\alpha^{(2)}(x)\right\},
	\end{equation}
	where
	\begin{equation}\label{twosolutions}
		f_\alpha^{(1)}(x) = \frac{1}{x}, \quad f_\alpha^{(2)}(x) = \frac{1}{2} \left(\alpha(2 - x) + \sqrt{\alpha^2(2 - x)^2 + 4 (1 - \alpha)}\right).
	\end{equation}
\end{proposition}

Depending on the value of $\alpha$, there are three different scenarios:
\begin{itemize}
	\item $\alpha = 1$: In this case, the limit function is simply $f_\infty(x) = \frac{1}{x}$.  
	\item $\frac{1}{2} < \alpha < 1$: In this case, there is a transition point $x_* = \frac{1}{2\alpha - 1}$. The limit function $f_\infty(x)$ equals $f_\alpha^{(1)}(x)$ for $x \geqslant x_*$ and $f_\alpha^{(2)}(x)$ for $1 \leqslant x \leqslant x_*$. The function is continuous, but its first derivative is discontinuous at $x = x_*$. 
	\item $0 < \alpha \leqslant \frac{1}{2}$: In this case, the limit function is $f_\infty = f_\alpha^{(2)}$. 
\end{itemize}

The proof of proposition \ref{prop:fsequence} is organized as follows. Section \ref{sec:EulerLagrangian} introduces the fluid description of \eqref{f:recursion0}, which significantly simplifies the problem. Using the fluid description, we prove parts (a)-(c) of the proposition in section \ref{section:proofabc}. In section \ref{sec:limit}, we prove part (d), where we divide the discussion into the three cases mentioned above.

\subsubsection{From Eulerian to Lagrangian Description}\label{sec:EulerLagrangian}

The iteration equation \eqref{f:recursion0} involves the term $f_k\left(\frac{1}{f_{k+1}(x)}\right)$, which complicates both analytical and numerical analysis. To simplify the problem, we introduce an equivalent version of \eqref{f:recursion0}. This approach is inspired by classical fluid dynamics. In this analogy, we interpret $k$ as time, $x$ as position, and $f_k(x)$ as a velocity field evaluated at the ``spacetime" point $(k, x)$. Thus, \eqref{f:recursion0} resembles the equation of motion in the Eulerian description.

In contrast, the Lagrangian description involves tracking individual fluid particles as they move through space and time. We define the sequence $\{z_k\}$ as follows:
\begin{equation}\label{def:z}
	\begin{split}
		z_k &= \max\left\{ y \geqslant 1 \,\Big{|}\, f_{k}(y) = \frac{1}{z_{k-1}} \right\} \quad (k \geqslant 1), \\
		z_0 &= x, \quad z_{-1} = 1,\quad(x\geqslant1). \\
	\end{split}
\end{equation}
We use ``max" in the definition of $z_k$ because the strict monotonicity of $f_k$ has not yet been established (in fact, even the well-definedness of $z_k$ has not yet been justified). Here, $z_k$ can be interpreted as the position of a fluid particle at time $k$, starting from the initial position $x$.

Using the Lagrangian description, \eqref{f:recursion0} and \eqref{def:Fk} are equivalent to a second-order difference equation:
\begin{equation}\label{eq:recur3}
	\begin{split}
		&\alpha z_{k+1} - z_k + \frac{1}{z_k} + (1 - \alpha)z_{k-1} - \frac{1}{z_{k-1}} = 0 \quad (k \geqslant 0), \\
		&z_0 = x, \quad z_{-1} = 1 \quad (x \geqslant 1).
	\end{split}
\end{equation}
We see that the iteration equation is significantly simplified in the Lagrangian description. Eq.\,\eqref{eq:recur3} can be interpreted as the equation of motion for non-interacting fluid particles, as particles with different initial conditions do not interact with each other in this equation.

To recover the Eulerian description, the following reparametrization for the function $f_k$ is used:
\begin{equation}\label{fk:reparametrization}
	\begin{split}
		X = z_k(x), \quad Y = \frac{1}{z_{k-1}(x)}, \quad Y = f_k(X) \quad (x \geqslant 1).
	\end{split}
\end{equation}
It is not immediately clear that $X = z_k(x)$ provides a valid reparametrization in each iteration step, as it has not been established that $z_k(x)$ is a one-to-one map from $[1, \infty)$ to itself. Later, we will show that this is indeed the case.

From \eqref{eq:recur3}, we observe that the fluid particle has a conserved quantity given by:
\begin{equation}\label{conservation}
	Q_k := \alpha z_{k} - (1 - \alpha)z_{k-1} + \frac{1}{z_{k-1}} \equiv \alpha(x + 1).
\end{equation}
We will see later that the conservation law \eqref{conservation} is crucial for studying the limits of the sequences $\{z_k\}$ and $\{f_k\}$.

\subsubsection{Proof of proposition \ref{prop:fsequence}-(a) (b) (c)}\label{section:proofabc}
In this subsection, we present the proof of proposition \ref{prop:fsequence}-(a), (b), and (c). Since these parts contain multiple statements, we will split the proof into several smaller lemmas and prove them one by one. First, we will establish some properties of the sequence $\{z_k\}$ defined in \eqref{eq:recur3}. Then, we will use these properties to prove the properties of the function sequence $\{f_k\}$. Throughout this subsection, we assume that $\alpha \in (0, 1]$.

From \eqref{eq:recur3}, for each initial condition $(z_{-1}, z_0)$, the sequence $\{z_k\}$ is uniquely determined (since $\alpha > 0$). The well-definedness of $\{z_k\}$ is equivalent to proving that $z_k \neq 0$ for each $k$. This will be shown below.

Let $z_k(x)$ denote the sequence $\{z_k\}$ with the initial conditions $z_{-1} = 1$ and $z_0 = x$. Then, the following lemma holds:

\begin{lemma}\label{lemma:zk}
	For any $k \geqslant 0$, $z_k(x)$ is a well-defined, strictly increasing smooth function of $x$. Additionally, we have $z_k(1) = 1$ and $z_k(\infty) = \infty$.
\end{lemma}

\begin{proof}
	We begin by proving that $z_k(1) = 1$. This holds for $k = 0$ by definition. Using induction, assume that $z_k(1) = 1$ for all $k \leqslant n$. For $k = n+1$, using \eqref{conservation}, we have:
	\begin{equation}
		\begin{split}
			z_{n+1}(1) = \frac{1}{\alpha} \left[z_n(1) - \frac{1}{z_n(1)} - (1 - \alpha) z_{n-1}(1) + \frac{1}{z_{n-1}(1)} \right].
		\end{split}
	\end{equation}
	Since $z_n(1) = 1$, this implies that $z_{n+1}(1) = 1$. By induction, we conclude that $z_k(1) = 1$ for all $k \geqslant 0$.
	
	Next, we prove inductively that $z'_k(x) \geqslant 1$ for all $k \geqslant 0$ and $x \geqslant 1$. For $k = 0$, we have $z'_0(x) = 1$ by definition. Now assume that $z'_k(x) \geqslant 1$ for $k \leqslant n$. For $k = n+1$, taking the $x$-derivative of \eqref{conservation} gives:
	\begin{equation}
		\begin{split}
			\alpha z_{n+1}'(x) - z_n'(x) \left[1 - \alpha + \frac{1}{z_n(x)^2}\right] = \alpha.
		\end{split}
	\end{equation}
	Thus, we obtain:
	\begin{equation}
		\begin{split}
			z_{n+1}'(x) = 1 + \frac{z'_n(x)}{\alpha} \left[1 - \alpha + \frac{1}{z_n(x)^2}\right] \geqslant 1.
		\end{split}
	\end{equation}
	This completes the proof of $z'_k(x) \geqslant 1$.
	
	By $z_k'(x)\geqslant1$ and the fact that $z_k(1) = 1$, we conclude that $z_k(x) \geqslant x$ for all $x \geqslant 1$ in each iteration step. Therefore, $z_k(x) \neq 0$ for all $x \geqslant 1$ and $z_k(\infty) = \infty$. In particular, this shows that the sequence $\{z_k\}$ in \eqref{eq:recur3} is well-defined.
	
	Finally, we show that $z_k(x)$ is smooth for each fixed $k \geqslant 0$. By the conservation law \eqref{conservation}, $z_k(x)$ and $z_{k-1}(x)$ are related by:
	\begin{equation}
		\begin{split}
			z_k(x) = \frac{1}{\alpha} \left[(1 - \alpha) z_{k-1}(x) - \frac{1}{z_{k-1}(x)}\right] + x + 1 \quad (x \geqslant 1).
		\end{split}
	\end{equation}
	Since $z_k(x) \neq 0$ for all $k \geqslant 0$ and $x \geqslant 1$, and since $z_0(x) = x$ is smooth in $x$, the above expression shows the smoothness of $z_k(x)$ by induction. This completes the proof of the lemma.
\end{proof}
Now it becomes obvious that $f_k(x)$ is well-defined in each iteration step of \eqref{f:recursion0}. Using \eqref{def:z} and lemma \ref{lemma:zk}, we can reparametrize the function $f_k(x)$ as in \eqref{fk:reparametrization}. Thus, we conclude:

\begin{lemma}\label{lemma:fk}
	For all $k \geqslant 1$, $f_k(x)$ is a strictly decreasing smooth function on $[1, \infty)$. Additionally, we have $f_k(1) = 1$ and $f_k(\infty) = 0$.
\end{lemma}
\begin{proof}
	By \eqref{def:z} and lemma \ref{lemma:zk}, we can express $f_k(x)$ as
	\begin{equation}
		\begin{split}
			f_k(x) = \frac{1}{z_{k-1} \circ z_{k}^{-1}(x)}.
		\end{split}
	\end{equation}
	Here, $z_{k}^{-1}$ denotes the inverse of $z_k$. By lemma \ref{lemma:zk}, for $k \geqslant 1$, both $z_{k-1}$ and $z_k^{-1}$ are strictly increasing, one-to-one, smooth functions on $[1, \infty)$. Therefore, $f_k(x)$ is a strictly decreasing smooth function on $[1, \infty)$.
	
	Furthermore, by lemma \ref{lemma:zk}, for $k \geqslant 1$, we have
	\begin{equation}
		\begin{split}
			f_k(1) = \frac{1}{z_{k-1}\left(z_k^{-1}(1)\right)} = \frac{1}{z_{k-1}(1)} = \frac{1}{1} = 1,
		\end{split}
	\end{equation}
	and
	\begin{equation}
		\begin{split}
			f_k(\infty) = \frac{1}{z_{k-1}\left(z_k^{-1}(\infty)\right)} = \frac{1}{z_{k-1}(\infty)} = \frac{1}{\infty} = 0.
		\end{split}
	\end{equation}
	This completes the proof of the lemma.
\end{proof}

To complete the proof of proposition \ref{prop:fsequence}-(a), we need the following additional lemma:

\begin{lemma}\label{lemma:fklowerbound}
	$f_k(x) > \frac{1}{x}$ for any $k \geqslant 0$ and $x > 1$.
\end{lemma}
\begin{proof}
	For $f_0(x)$, this is obvious from \eqref{f:recursion0}. By induction, suppose the lemma holds for $k \leqslant n$. For $x > 1$, we know that $y = 1$ is not the solution to $F_n(y) = \alpha x$ since $F_n(1) = \alpha$. Then, for $y < 1$, we have:
	\begin{equation}
		\begin{split}
			F_n(y) > \frac{1}{y} - y - \frac{1 - \alpha}{y} + y = \frac{\alpha}{y},
		\end{split}
	\end{equation}
	which implies that the solution to $F_n(y) = \alpha x$ must satisfy
	\begin{equation}
		\begin{split}
			y > \frac{1}{x}.
		\end{split}
	\end{equation}
	Hence, we get $f_{n+1}(x) > \frac{1}{x}$ for all $x > 1$.
\end{proof}

Now proposition \ref{prop:fsequence}-(a) is proven.

As a consequence of lemma \ref{lemma:fk}, we can remove ``max" in the definitions of $\left\{z_k\right\}$ in \eqref{def:z}, because the solution to $f_k(y) = \frac{1}{z_{k-1}}$ is now unique for any $k\geqslant1$.

As a direct consequence of lemma \ref{lemma:fklowerbound}, we have:
\begin{equation}\label{z:mono:k}
	\begin{split}
		z_{k+1}(x) \geqslant z_k(x), \quad \forall k \geqslant 0 \ \text{and}\ x \geqslant 1\,,
	\end{split}
\end{equation}
i.e., with fixed initial conditions, the sequence $\{z_k\}$ is monotonically increasing. This property will be useful later when we examine the limit of the function sequence $\{f_k\}$.

Next, we prove proposition \ref{prop:fsequence}-(b), which is summarized in the following lemma:

\begin{lemma}\label{lemma:FFk}
	For all $k \geqslant 0$, $F_k(y)$ is a strictly decreasing smooth function from $(0, 1]$ to $[\alpha, \infty)$. Additionally, we have $F_k(0) = \infty$ and $F_k(1) = \alpha$.
\end{lemma}
\begin{proof}
	Using \eqref{def:Fk} and lemma \ref{lemma:fk}, for $k \geqslant 0$, we can take $y = f_{k+1}(x)$ and rewrite $F_k$ as:
	\begin{equation}
		\begin{split}
			F_k(y) = \alpha x = \alpha f^{-1}_{k+1}(y).
		\end{split}
	\end{equation}
	Here, $f^{-1}_{k+1}$ denotes the inverse of $f_{k+1}$, which is a strictly decreasing, one-to-one, smooth function from $(0, 1]$ to $[1, \infty)$, as a consequence of lemma \ref{lemma:fk}. It follows that $F_k(y)$ is a strictly decreasing smooth function on $(0, 1]$.
	
	From lemma \ref{lemma:fk}, we deduce:
	\begin{equation}
		\begin{split}
			F_k(0) = \alpha f_{k+1}^{-1}(0) = \infty,
		\end{split}
	\end{equation}
	and
	\begin{equation}
		\begin{split}
			F_k(1) = \alpha f_{k+1}^{-1}(1) = \alpha.
		\end{split}
	\end{equation}
	This completes the proof of the lemma.
\end{proof}

Now proposition \ref{prop:fsequence}-(b) has been proven. Using lemma \ref{lemma:FFk}, we can also remove ``max'' from the definition of $\{f_k(x)\}$ in \eqref{f:recursion0}.

It remains to prove proposition \ref{prop:fsequence}:
\begin{lemma}\label{lemma:fkmono}
	We have $f_{k+1}(x) < f_{k}(x)$ for any $x > 1$.
\end{lemma}
\begin{proof}
	For $k = 0$, we have $f_{1}(x) < f_{0}(x)$ ($x > 1$) by the explicit forms:
	\begin{equation}
		\begin{split}
			f_0(x) = 1,\quad f_1(x) = \frac{1}{2} \left( \sqrt{\alpha^2 (x-1)^2 + 4} - \alpha (x-1) \right).
		\end{split}
	\end{equation}
	We prove the case for $k \geqslant 1$ by induction. Suppose the lemma holds for $k \leqslant n-1$. From \eqref{f:recursion0} and \eqref{def:Fk}, we have
	\begin{equation}\label{eq:diffrecur}
		\begin{split}
			&\frac{1}{f_{n+1}(x)} - f_{n+1}(x) - \frac{1 - \alpha}{f_n\left( \frac{1}{f_{n+1}(x)} \right)} + f_n\left( \frac{1}{f_{n+1}(x)} \right) \\
			&= \frac{1}{f_n(x)} - f_n(x) - \frac{1 - \alpha}{f_{n-1}\left( \frac{1}{f_n(x)} \right)} + f_{n-1}\left( \frac{1}{f_n(x)} \right).
		\end{split}
	\end{equation}
	Consider $x > 1$. By lemma \ref{lemma:fklowerbound}, we have $\frac{1}{f_{n+1}(x)} > 1$. Then by the induction, we have
	\begin{equation}
		\begin{split}
			f_n\left( \frac{1}{f_{n+1}(x)} \right) < f_{n-1}\left( \frac{1}{f_{n+1}(x)} \right).
		\end{split}
	\end{equation}
	This inequality, combined with \eqref{eq:diffrecur}, implies that
	\begin{equation}
		\begin{split}
			&\frac{1}{f_{n+1}(x)} - f_{n+1}(x) - \frac{1 - \alpha}{f_{n-1}\left( \frac{1}{f_{n+1}(x)} \right)} + f_{n-1}\left( \frac{1}{f_{n+1}(x)} \right) \\
			&> \frac{1}{f_n(x)} - f_n(x) - \frac{1 - \alpha}{f_{n-1}\left( \frac{1}{f_n(x)} \right)} + f_{n-1}\left( \frac{1}{f_n(x)} \right),
		\end{split}
	\end{equation}
	for any $x > 1$. Using \eqref{def:Fk}, the above inequality can be rewritten as:
	\begin{equation}
		\begin{split}
			F_{n-1}(f_{n+1}(x)) > F_{n-1}(f_n(x)), \quad (x > 1).
		\end{split}
	\end{equation}
	Then, by lemma \ref{lemma:FFk}, we get $f_{n+1}(x) < f_n(x)$ for all $x > 1$.
\end{proof}

This completes the proof of proposition \ref{prop:fsequence}-(a) (b) (c). Part (a) follows from lemmas \ref{lemma:fk} and \ref{lemma:fklowerbound}, part (b) follows from lemma \ref{lemma:FFk}, and part (c) follows from lemma \ref{lemma:fkmono}.

\subsubsection{The limit of the sequence: proof of proposition \ref{prop:fsequence}-(d)}\label{sec:limit}
Proposition \ref{prop:fsequence}-(a) and (c) imply that the function sequence $\left\{ f_k \right\}$ has a pointwise limit \( f_\infty \). In this section, we aim to find the explicit expression of this limit.

As demonstrated in the previous section, we can remove ``max" in the definition \eqref{f:recursion0} by the monotonicity of \( F_k \). Thus, \eqref{f:recursion0} can be rewritten as:
\begin{equation}\label{eq:recur2}
	\begin{split}
		&f_0(x) \equiv 1, \\
		&\alpha x + f_{k+1}(x) - \frac{1}{f_{k+1}(x)} + \frac{1 - \alpha}{f_k\left( \frac{1}{f_{k+1}(x)} \right)} - f_k\left( \frac{1}{f_{k+1}(x)} \right) = 0, \quad k \geqslant 0. \\
	\end{split}
\end{equation}
The limit function, denoted as \( f_\infty(x) \), satisfies the fixed-point equation of \eqref{eq:recur2}:
\begin{equation}\label{f:fp}
	\alpha x + f_{\infty}(x) - \frac{1}{f_{\infty}(x)} + \frac{1 - \alpha}{f_\infty\left( \frac{1}{f_{\infty}(x)} \right)} - f_\infty\left( \frac{1}{f_{\infty}(x)} \right) = 0, \quad x \geqslant 1.
\end{equation}
The function \( f(x) = \frac{1}{x} \) solves the fixed-point equation \eqref{f:fp}. However, this does not necessarily mean it is the correct limit of our iteration process, as there might be other solutions to \eqref{f:fp}. Intuitively, if a fixed point exists that is larger than \( \frac{1}{x} \), it is more likely to be the limit function since the sequence \( f_{k}(x) \) is monotonically decreasing in \( k \). Indeed, we find that there are two solutions to \eqref{f:fp}, given in \eqref{twosolutions}. Both solve \eqref{f:fp} and satisfy several properties expected of the limit of our iteration procedure.\footnote{In this paper, we do not prove that \eqref{twosolutions} are the only solutions to the fixed-point equation \eqref{f:fp} under certain constraints. However, we will show that these two solutions are the relevant ones for describing the limit of the iteration process \eqref{eq:recur2}.} So it requires further analysis to decide which should be the correct limit. This will be done in this subsection.
\newline
\newline \textbf{Case 1: $\alpha=1$}
\newline
\newline When $\alpha = 1$, the conservation law \eqref{conservation} simplifies to:
\begin{equation}
	z_{k+1} + \frac{1}{z_k} = x + 1.
\end{equation}
Given that $z_0 = x$, the solution to this difference equation is:
\begin{equation}
	\begin{split}
		z_n(x) = &\Bigg[\left((x+2)\sqrt{x-1} + x\sqrt{x+3}\right) \left(\sqrt{x+3}+\sqrt{x-1}\right)^{2n} \\
		&- \left((x+2)\sqrt{x-1} - x\sqrt{x+3}\right) \left(\sqrt{x+3} - \sqrt{x-1}\right)^{2n}\Bigg] \\
		&\times \Bigg[\left(\sqrt{x-1} + \sqrt{x+3}\right)^{2n+1} + \left(\sqrt{x+3} - \sqrt{x-1}\right)^{2n+1} \Bigg]^{-1}.
	\end{split}
\end{equation}
It can be verified that $z_n(1) = 1$ and $z_n(\infty)=\infty$ for any $n\geqslant0$, as expected according to lemma \ref{lemma:fk}. In the limit $n \to \infty$ with fixed $x$, the function $z_n(x)$ converges to a finite limit:
\begin{equation}
	z_{\infty}(x) = \frac{(x+2)\sqrt{x-1} + x\sqrt{x+3}}{\sqrt{x-1} + \sqrt{x+3}}.
\end{equation}
As $x$ increases from $1$ to $\infty$, $z_{\infty}(x)$ also ranges from $1$ to $\infty$. Returning to the Eulerian description, the above argument shows that the iteration process defined by \eqref{eq:recur2} converges to the pointwise limit:
\begin{equation}
	f_{\infty}(x) = \frac{1}{x} \quad (\alpha = 1).
\end{equation}
\newline
\newline \textbf{Case 2: $\frac{1}{2} < \alpha < 1$}
\newline
\newline When $\frac{1}{2} < \alpha < 1$, an analytic solution to \eqref{eq:recur3} is not readily available, but we can rewrite the equation in the following form:
\begin{equation}\label{ratio:difference}
	\frac{z_{k+1} - z_k}{z_k - z_{k-1}} = \frac{1 - \alpha + \frac{1}{z_k z_{k-1}}}{\alpha}.
\end{equation}
This expression indicates that the difference $z_{k+1} - z_k$ grows exponentially until $z_{k} > \frac{1}{\sqrt{2\alpha - 1}}$ at some iteration. By the monotonicity \eqref{z:mono:k} of the sequence $\left\{ z_k \right\}$, once $z_{k}$ exceeds this value, it remains greater than $\frac{1}{\sqrt{2\alpha - 1}}$, and the difference in the sequence starts to decay exponentially. Consequently, for any $x > 1$, the sequence $\{z_k\}$ converges to a limit determined by the conservation law \eqref{conservation}:
\begin{equation}
	z_\infty(x) = \frac{\alpha (x+1) + \sqrt{\alpha^2 (x+1)^2 - 4(2\alpha - 1)}}{4 \alpha - 2}.
\end{equation}
We observe that the function $z_{\infty}(x)$ is monotonically increasing in $x$, with the lower and upper bounds given by $\lim\limits_{x\rightarrow1^+}z_{\infty}(x) = \frac{1}{2\alpha - 1}$ and $z_{\infty}(\infty) = \infty$. Returning to the Eulerian description, we conclude that:
\begin{equation}\label{firstsolution}
	f_{\infty}(x) = \frac{1}{x} \quad \left( x \geqslant \frac{1}{2\alpha - 1} \right).
\end{equation}

Now, we determine the function $f_{\infty}(x)$ in the range $1 \leqslant x < \frac{1}{2\alpha - 1}$. From the previous argument, we know that for any fluid particle starting from $z_0 = x > 1$, it eventually ends at $z_\infty(x) > \frac{1}{2\alpha - 1}$. To capture the particles that end at $x \leqslant \frac{1}{2\alpha - 1}$, we need to adjust the initial condition at each iteration step.

Specifically, we take the limit \(k \to \infty\) with \(z_k(x_k) = x \in \left[1, \frac{1}{2\alpha-1}\right)\) fixed. Since \(z_k(x)\) is monotonically increasing in both \(x\) and \(k\), by lemma \ref{lemma:zk} and \eqref{z:mono:k}, this tuning leads to a monotonically decreasing sequence of \(x_k\):
\begin{equation}
	\begin{split}
		x_1 \geqslant x_2 \geqslant x_3 \geqslant \dots
	\end{split}
\end{equation}
Additionally, \(x_k \geqslant 1\), so the sequence \(\left\{x_k\right\}\) has a limit \(x_\infty\). We claim that \(x_\infty = 1\). Otherwise, suppose \(x_\infty > 1\), we then have
\begin{equation}\label{xk:contradiction}
	\begin{split}
		x = \lim\limits_{k\rightarrow\infty} z_{k}(x_k) \geqslant \lim\limits_{k\rightarrow\infty} z_{k}(x_\infty) \geqslant \frac{1}{2\alpha-1}.
	\end{split}
\end{equation}
In the first step, we use \(z_k(x_k) = x\), in the second step, we apply \(x_k \geqslant x_\infty\) and lemma \ref{lemma:zk}, and in the final step, we use the fact that for any initial condition \(z_0 = x > 1\), the limit satisfies \(z_\infty(x) \geqslant \frac{1}{2\alpha-1}\), as demonstrated above.

This shows that \eqref{xk:contradiction} contradicts our assumption that \(x \in \left[1, \frac{1}{2\alpha-1}\right)\). Hence, we must have \(x_\infty = 1\).

Using \eqref{conservation} with the fine-tuned initial condition \(z_0 = x_k\) for \(Q_k\), we obtain:
\begin{equation}\label{conservation:tuning}
	\begin{split}
		\alpha x - (1-\alpha) z_{k-1}(x_k) + \frac{1}{z_{k-1}(x_k)} = \alpha(x_k + 1).
	\end{split}
\end{equation}
Recalling that \(f_{k}(z_{k}) = \frac{1}{z_{k-1}}\) for the fixed initial condition, and since we now fix \(z_k(x_k) = x\), we have:
\begin{equation}
	\begin{split}
		f_{k}(x) \equiv f_{k}(z_{k}(x_{k})) = \frac{1}{z_{k-1}(x_{k})}.
	\end{split}
\end{equation}
Taking the limit \(k \rightarrow \infty\), we find:
\begin{equation}\label{limit:zk-1}
	\begin{split}
		\lim\limits_{k\rightarrow\infty} z_{k-1}(x_{k}) = \frac{1}{f_\infty(x)}.
	\end{split}
\end{equation}
Now, taking the limit \(k \rightarrow \infty\) in \eqref{conservation:tuning}, and using \(x_k \rightarrow 1\) along with \eqref{limit:zk-1}, we get:
\begin{equation}
	\alpha x - \frac{1 - \alpha}{f_\infty(x)} + f_\infty(x) = 2\alpha, \quad x \in \left[1, \frac{1}{2\alpha-1}\right).
\end{equation}
This equation has two solutions for \(f_\infty(x)\), but one solution is always below \(\frac{1}{x}\), which contradicts lemma \ref{lemma:fklowerbound}, so we discard it. The relevant solution is:
\begin{equation}\label{secondsolution}
	f_\infty(x) = \frac{1}{2} \left(\alpha(2 - x) + \sqrt{\alpha^2(2 - x)^2 + 4 (1 - \alpha)}\right).
\end{equation}
We can explicitly verify that \eqref{secondsolution} satisfies the fixed-point equation \eqref{f:fp}, and it yields:
\begin{equation}
	f_\infty(1) = 1, \quad f_\infty\left(\frac{1}{2\alpha - 1}\right) = 2\alpha - 1.
\end{equation}
Therefore, the two solutions \eqref{firstsolution} and \eqref{secondsolution} merge continuously at \(x = \frac{1}{2\alpha-1}\), resulting in:
\begin{equation}\label{case2:solution}
	f_\infty(x) = 
	\begin{cases}
		\dfrac{1}{x}, & x \geqslant \frac{1}{2\alpha - 1}, \\
		& \\
		\dfrac{1}{2} \left(\alpha(2 - x) + \sqrt{\alpha^2(2 - x)^2 + 4 (1 - \alpha)}\right), & 1 \leqslant x \leqslant \frac{1}{2\alpha - 1}. \\ 
	\end{cases}
\end{equation}
We confirm \eqref{case2:solution} numerically by performing many iterations. See figure \ref{fig:case2numerics}.
\begin{figure}[ht]
	\begin{minipage}{0.48\textwidth}
		\centering
		\includegraphics[width=\linewidth]{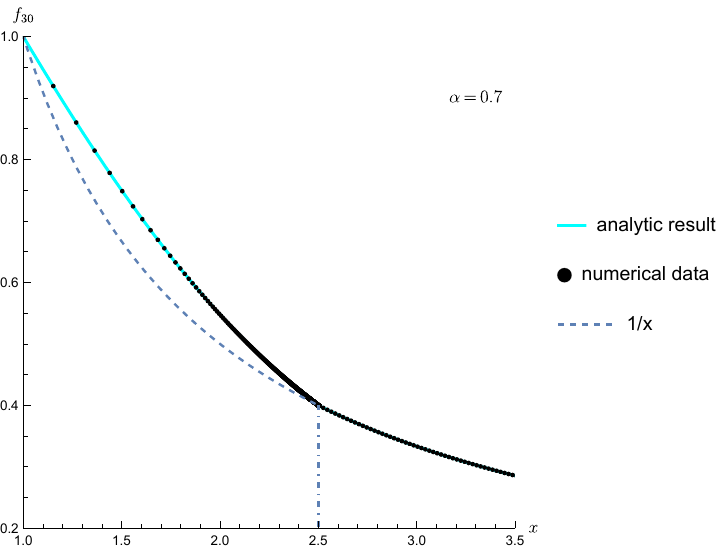}
	\end{minipage}\hfill
	\begin{minipage}{0.48\textwidth}
		\centering
		\includegraphics[width=\linewidth]{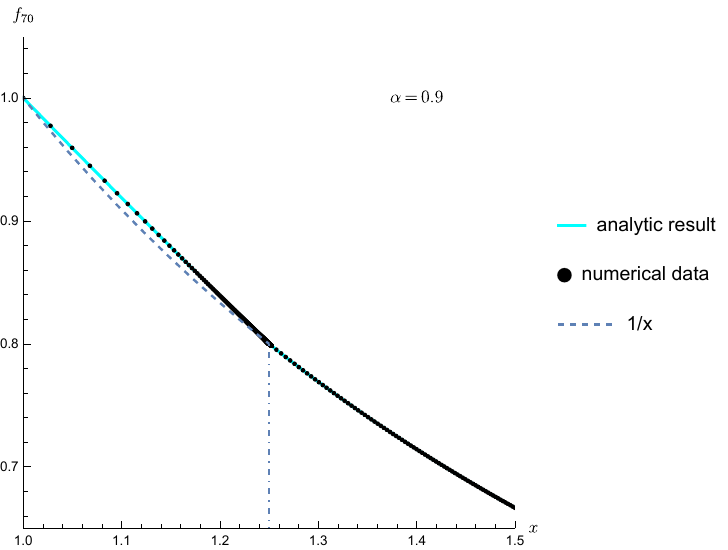}
	\end{minipage}
	\vspace{0.5cm} 
	\begin{minipage}{0.48\textwidth}
		\centering
		\vspace{0.5cm} 
		\includegraphics[width=\linewidth]{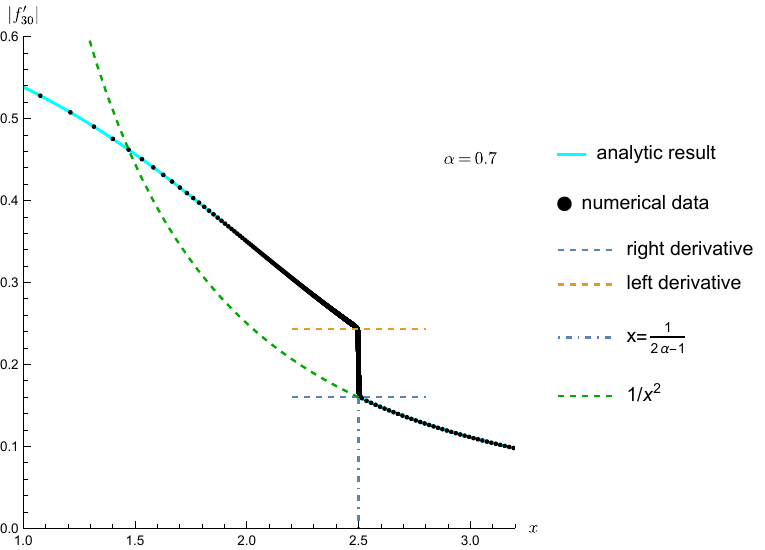}
	\end{minipage}\hfill
	\begin{minipage}{0.48\textwidth}
		\centering
		\vspace{0.5cm} 
		\includegraphics[width=\linewidth]{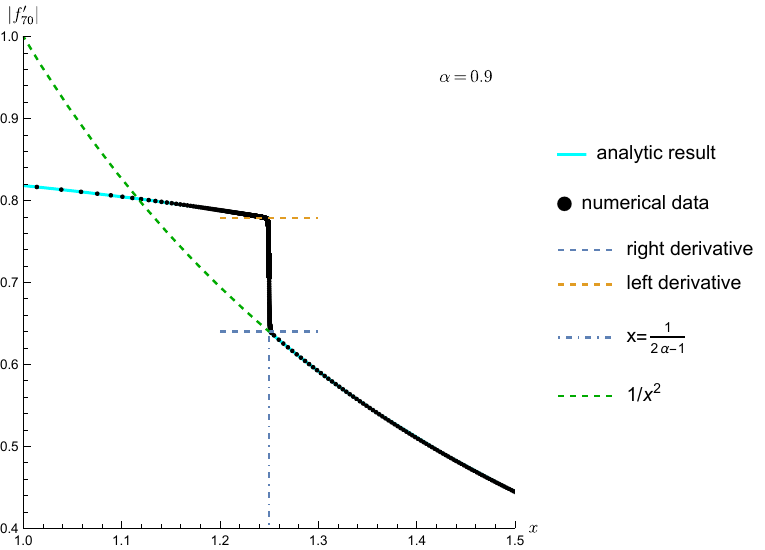}
	\end{minipage}
	\caption{Numerical tests of $f_\infty(x)$ (first row) and its derivative (second row) at $\alpha=0.7$ (30 iterations) and $\alpha=0.9$ (70 iterations). The black dots are numerical data after finite iterations. The light-blue curves are the analytic predictions of $f_\infty(x)$ given in \eqref{case2:solution}. The vertical blue dashed lines are the analytic prediction of the transition point $x_*=\frac{1}{2\alpha-1}$. In the second row, the horizonal dashed lines are the analytic prediction of the left- and right-derivatives at the transition point.}
	\label{fig:case2numerics}
\end{figure}
\newline
\newline
\noindent\textbf{Case 3: $0 < \alpha \leqslant \frac{1}{2}$}
\newline
\newline When $0 < \alpha \leqslant \frac{1}{2}$, it follows from \eqref{ratio:difference} that the difference $z_{k+1} - z_k$ never decreases:
\begin{equation}
	\frac{z_{k+1} - z_k}{z_k - z_{k-1}} = \frac{1 - \alpha + \frac{1}{z_k z_{k-1}}}{\alpha} \geqslant \frac{1 - \alpha}{\alpha} \geqslant 1.
\end{equation}
Therefore, the fluid particle starting from any fixed $x > 1$ eventually diverges to infinity:
\begin{equation}
	z_{\infty}(x) = \infty, \quad \forall x > 1 \ \text{fixed}.
\end{equation}
This situation is very similar to $f_\infty(x)$ in the regime $1 \leqslant x \leqslant \frac{1}{2\alpha - 1}$ when $\frac{1}{2} < \alpha < 1$: we need to take the limit $k \rightarrow \infty$ with $z_k(x_k) = x$ fixed. Following the same approach, we obtain the limit function:
\begin{equation}\label{case3:solution}
	f_\infty(x) = \frac{1}{2} \left(\alpha(2 - x) + \sqrt{\alpha^2(2 - x)^2 + 4 (1 - \alpha)}\right) \quad (x \geqslant 1).
\end{equation}
We also confirm \eqref{case3:solution} numerically by performing many iterations. See figure \ref{fig:case3numerics}.
\begin{figure}[ht]
	\begin{minipage}{0.26\textwidth}
		\centering
		\includegraphics[width=\linewidth]{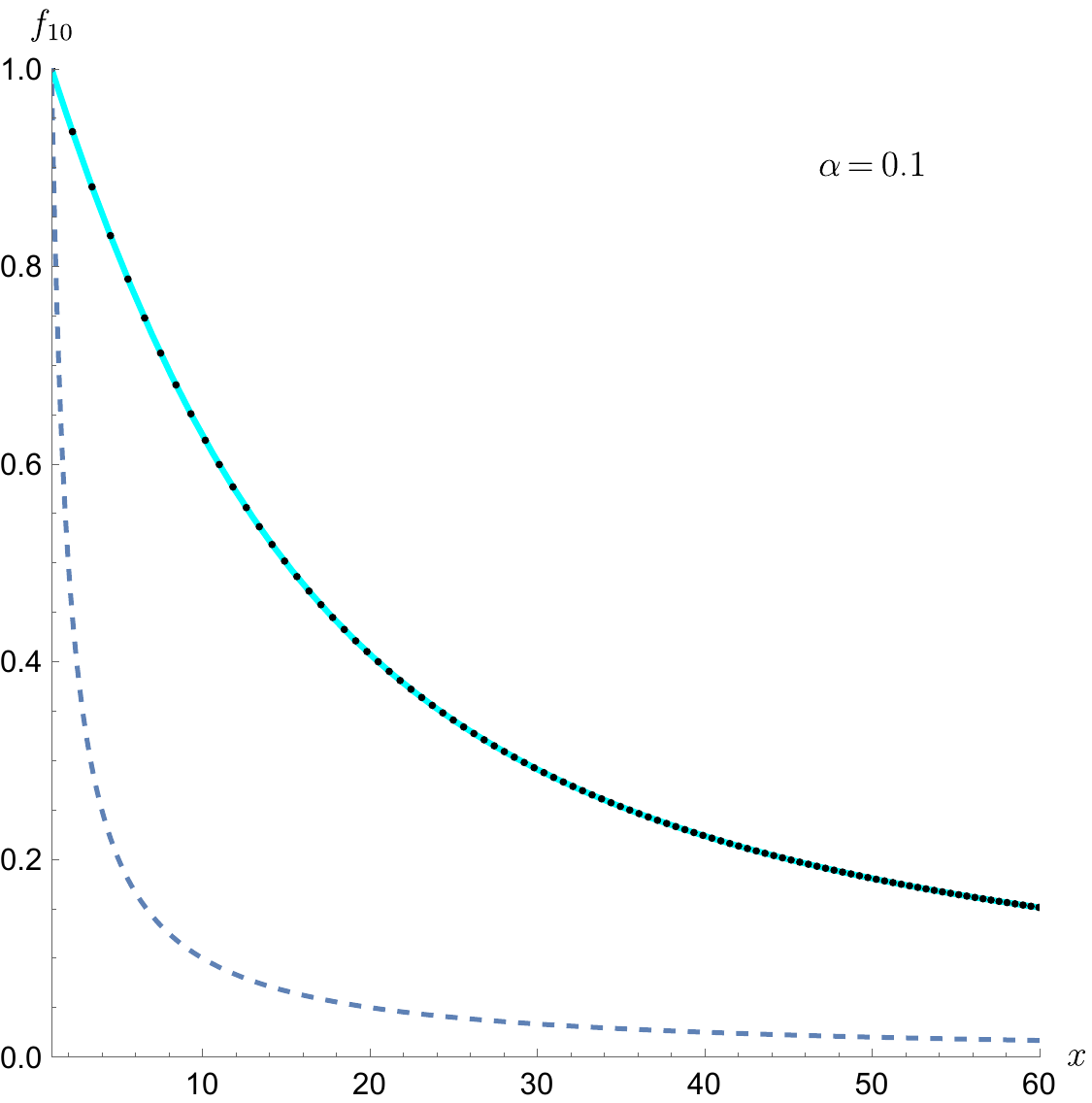}
	\end{minipage}\hfill
	\begin{minipage}{0.26\textwidth}
		\centering
		\includegraphics[width=\linewidth]{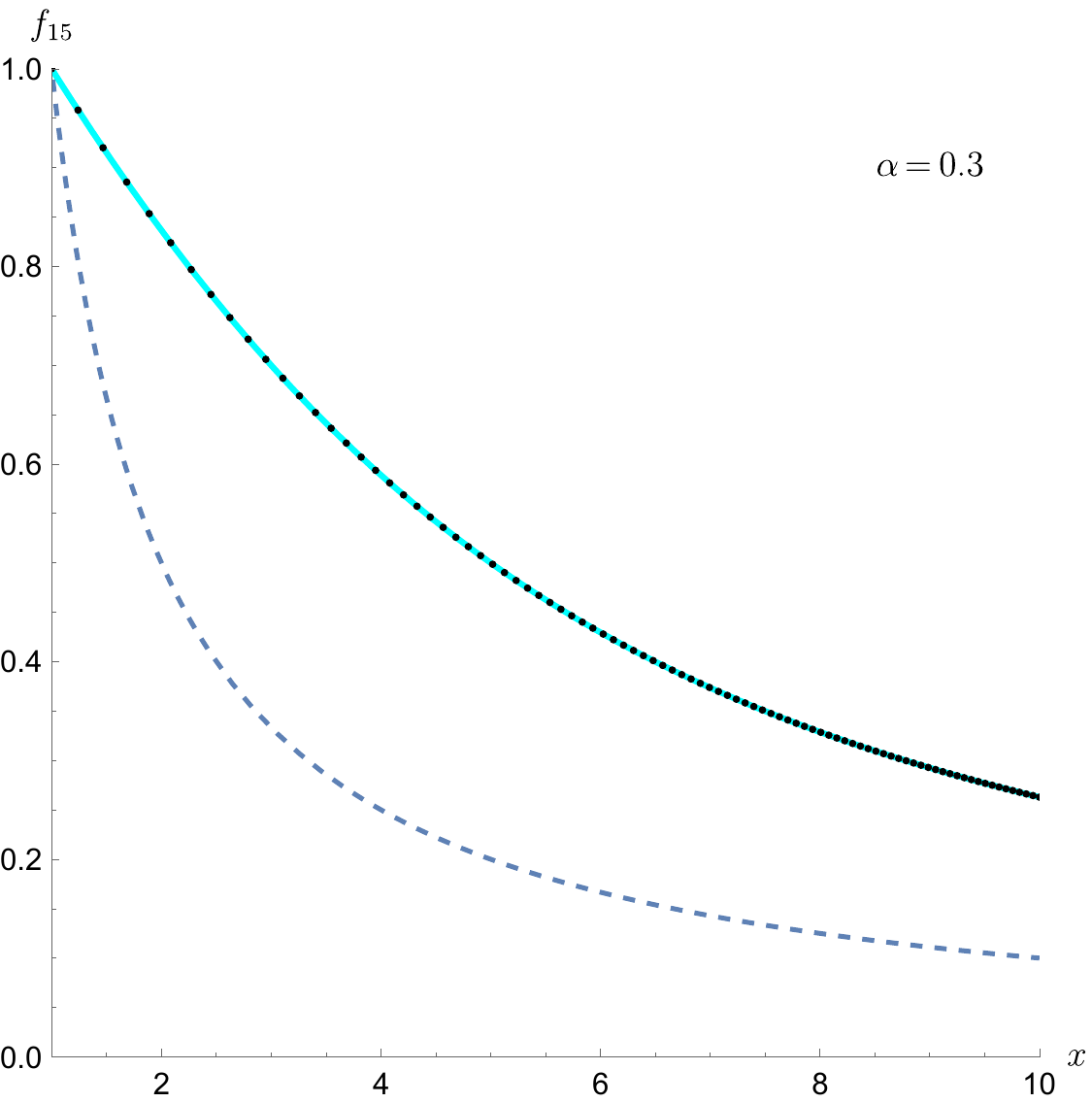}
	\end{minipage}\hfill
	\begin{minipage}{0.34\textwidth}
		\centering
		\includegraphics[width=\linewidth]{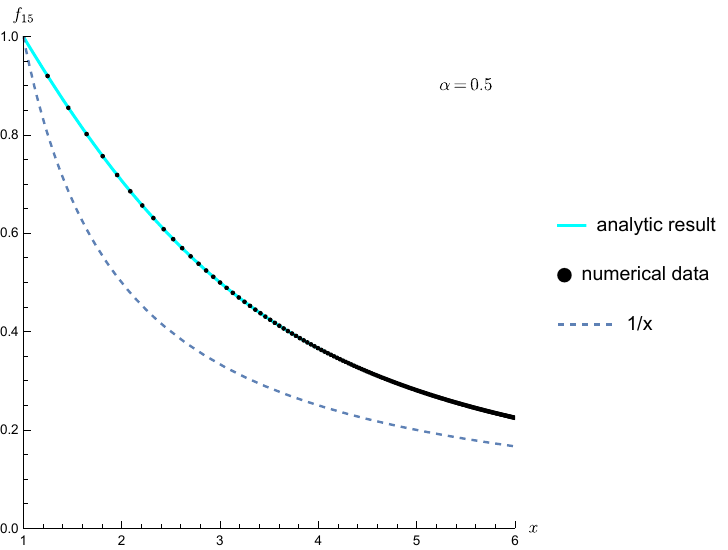}
	\end{minipage}
	\caption{Numerical tests of $f_\infty(x)$ at $\alpha=0.1$ (10 iterations), $\alpha=0.3$ (15 iterations) and $\alpha=0.5$ (15 iterations). The black dots are numerical data after finite iterations. The light-blue curves are the analytic predictions of $f_\infty(x)$ given in \eqref{case3:solution}.}
	\label{fig:case3numerics}
\end{figure}

\subsection{Step 2: proof by induction}\label{section:derivationiteration}
Having established the function sequence $\left\{f_k\right\}$, now we are ready to finish the proof of lemma \ref{lemma:functionbound}.

For $(x,y)$ in the regime $x,y>1$, the inequality \eqref{inequality:general} trivially holds by choosing
\begin{equation}
	\begin{split}
		N=1,\quad x_1=x,\quad y_1=y.
	\end{split}
\end{equation}
For $(x,y)$ in $\Omega_\alpha\backslash\{x,y>1\}$, we show inequality \eqref{inequality:general} using the following inductive argument. 
\begin{lemma}\label{lemma:iteration}
	Let $\alpha \in (0,1]$ be fixed. Suppose the function satisfies all the assumptions in section \ref{section:lemma} and suppose inequality \eqref{inequality:general} holds for $(x,y)$ in the regime
	\begin{equation}\label{def:Omegak}
		\Omega_\alpha^{(k)} = \left\{(x,y) \,\Big{|}\, y > 1,\ x > f_k(y);\ {\rm or}\ x > 1,\ y > f_k(x) \right\},
	\end{equation}
	with $N\leqslant N_k$. Here the function $f_k(x)$ was defined in \eqref{f:recursion0}.
	
	Then for each $(x,y)\in\Omega_\alpha^{(k+1)}$, the inequality \eqref{inequality:general} holds with some $N\leqslant N_k+1$.
\end{lemma}
\begin{proof}
	According to proposition \ref{prop:fsequence}-(a), the domain $\Omega_\alpha^{(k+1)}\backslash\Omega_\alpha^{(k)}$ consists of two disconnected components
	\begin{equation}
		\begin{split}
			{\rm I}.&\quad y>1,\quad f_{k+1}(y)<x\leqslant f_{k}(y), \\
			{\rm II}.&\quad x>1,\quad f_{k+1}(x)<y\leqslant f_{k}(x). \\
		\end{split}
	\end{equation}
	Since the arguments for two parts are similar, it suffices to show that \eqref{inequality:general} holds for component I, i.e.
	\begin{equation}\label{xyregime1}
		\begin{split}
			y>1,\quad f_{k+1}(y)<x\leqslant f_{k}(y).
		\end{split}
	\end{equation}
	We decompose the function $g(x,y)$ into
	$$g(x,y)=g_0(x,y)+g_1(x,y)$$
	as demonstrated in section \ref{section:lemma}. By proposition \ref{prop:fsequence}-(a), we have
	\begin{equation}
		\begin{split}
			x>f_{k+1}(y)>\frac{1}{y}.
		\end{split}
	\end{equation}
	This allows us to choose
	\begin{equation}
		\begin{split}
			x_1=x,\quad y_1=y
		\end{split}
	\end{equation}
	for inequality \eqref{inequality:general}. So it remains to show that $g_1(x,y)$ is bounded from above by the remaining terms in the right-hand side of \eqref{inequality:general}.
	
	Let us first consider the case of $k\geqslant1$ since the case of $k=0$ requires some extra justification.
	
	We introduce an auxiliary variable $y'$ within the range:
	\begin{equation}\label{range:y'}
		\begin{split}
			0<y'\leqslant y.
		\end{split}
	\end{equation}
	In this regime, $g_1(x,y)$ is bounded from above by:
	\begin{equation}\label{boundg1:step1}
		\begin{split}
			g_1(x,y)&\leqslant e^{-\alpha\gamma(y-y')}g_1(x,y') \\
			&\leqslant e^{-\alpha\gamma(y-y')}g(x,y') \\
			&=e^{-\alpha\gamma(y-y')} e^{\gamma\left(\frac{1}{x}+\frac{1}{y'}-x-y'\right)}g\left(\frac{1}{x},\frac{1}{y'}\right).
		\end{split}
	\end{equation}
	In the first line, we use assumption \eqref{twistprop}; in the second line, we use the assumption that both $g_0$ and $g_1$ are non-negative; in the last line, we use the modular invariance condition \eqref{g:modinv}.
	
	By proposition \ref{prop:fsequence}-(c), we have
	\begin{equation}
		\begin{split}
			\frac{1}{x}\geqslant\frac{1}{f_k(y)}>\frac{1}{f_0(y)}=1,\quad k\geqslant1
		\end{split}
	\end{equation}
	which is within the domain of $f_k$. We can choose
	\begin{equation}\label{choice:y'}
		\begin{split}
			\frac{1}{y'}>f_k\left(\frac{1}{x}\right).
		\end{split}
	\end{equation}
	Any $y'$ satisfying \eqref{choice:y'} will also satisfy \eqref{range:y'} because
	\begin{itemize}
		\item $f_k$ is positive, implying that $y'$ is also positive;
		\item $f_k\left(\frac{1}{x}\right)>x$ by proposition \ref{prop:fsequence}-(a), implying that $y'<\frac{1}{x}<y$. 
	\end{itemize}
	Now by the assumption of the lemma, the inequality \eqref{inequality:general} holds with the above choice of $y'$, this allows us to continue the estimate \eqref{boundg1:step1}:
	\begin{equation}\label{boundg1:step2}
		\begin{split}
			g_1(x,y)&\leqslant e^{-\alpha\gamma(y-y')} e^{\gamma\left(\frac{1}{x}+\frac{1}{y'}-x-y'\right)}g\left(\frac{1}{x},\frac{1}{y'}\right) \\
			&\leqslant e^{-\gamma\left(\alpha y+x-\frac{1}{x}+(1-\alpha)y'-\frac{1}{y'}\right)} \left[\sum\limits_{i=1}^{N}g_0\left(x_i,y_i\right)+g_1(x_{N},y_{N})\right]. \\
		\end{split}
	\end{equation}
	Here $x_i$'s and $y_i$'s are determined by $\frac{1}{x}$, $\frac{1}{y'}$ and $\alpha$, and $N\leqslant N_k$ (they are independent of $g$).
	
	It remains to show that for $(x,y)$ in regime \eqref{xyregime1}, we can always choose $y'$ within the range \eqref{choice:y'}, such that the exponential factor in \eqref{boundg1:step2} is equal to 1,which is equivalent to
	\begin{equation}\label{def:kappa}
		\begin{split}
			\kappa(x,y,y')\equiv\alpha y+x-\frac{1}{x}+(1-\alpha)y'-\frac{1}{y'}=0.
		\end{split}
	\end{equation}
	For $(x,y)$ in the regime \eqref{xyregime1}, we have:
	\begin{equation}
		\begin{split}
			\sup\limits_{y'<\frac{1}{f_k(1/x)}}\kappa(x,y,y')=\alpha y-F_k(x).
		\end{split}
	\end{equation}
	Here, we use the monotonicity of $\kappa(x,y,y')$ in $y'$ and the definition \eqref{def:Fk} of $F_k$.
	
	Since $f_{k+1}(y)<x\leqslant f_{k}(y)<1$ (for $k\geqslant1$), we must have $F_k(x)<\alpha y$. Otherwise, suppose $F_k(x)\geqslant\alpha y$. By proposition \ref{prop:fsequence}-(b), we know that $F_k(1)=\alpha$ and $F_k$ is continuous. So there exists $\xi\in[x,1]$ that solves $F_k(\xi)=\alpha y$. Then by the definition \eqref{f:recursion0} of $f_{k+1}$ we get $f_{k+1}(y)\geqslant\xi\geqslant x$, which contradicts our assumption that $x>f_{k+1}(y)$. Therefore, we conclude that $\kappa(x,y,y')$ is positive for $y'$ sufficiently close to $\frac{1}{f_k(1/x)}$.
	
	On the otherhand, $\kappa(x,y,y')$ is negative for $y'$ sufficiently close to 0. By the continuity of $\kappa(x,y,y')$ in $y'$, there exists an $y'\in\left(0,\frac{1}{f_k(1/x)}\right)$ that solves \eqref{def:kappa}. This finishes the proof of the lemma for the case of $k\geqslant1$.
	
	Now we consider the special case $k=0$. The only reason that we take extra care of this case is that the regime
	\begin{equation}
		\begin{split}
			y>1,\quad f_1(y)<x\leqslant f_0(y)\equiv1
		\end{split}
	\end{equation}
	contains the points with $x=1$, which is not what we want for $g_1$ in the inequality \eqref{inequality:general}. However, for this case it suffices to show the lemma in the regime
	\begin{equation}\label{regime:k0}
		\begin{split}
			y>1,\quad f_1(y)<x<1.
		\end{split}
	\end{equation}
	This is because for any $x=1$, we can choose
	$x'=\frac{2}{\alpha(y-1)+2}$ and obtain the following bound on $g_1(x,y)$ using \eqref{twistprop}:
	\begin{equation}
		\begin{split}
			g_1(x,y)\leqslant e^{-\alpha\gamma(x-x')}g_1(x',y)\leqslant g_1(x',y).
		\end{split}
	\end{equation}
	One can explicitly check that $x'>f_1(y)$, so $(x',y)$ is still in domain $\Omega_\alpha^{(1)}$. Then the argument for the case $k\geqslant1$ applies here. This finishes the proof of the lemma.
\end{proof}
\begin{remark}\label{remark:iterationproof}
	\item (a) The argument in the proof of lemma \ref{lemma:iteration} is inspired by HKS \cite{Hartman:2014oaa}. There, they considered the special case of $\alpha=1$ and computed $f_0$, $f_1$ and $f_2$.\footnote{In \cite{Hartman:2014oaa}, the explicit exressions of $f_0$ and $f_1$ was given (in their convention). $f_3$ was computed numerically and only the plot was presented there, see figure 2 in \cite{Hartman:2014oaa}. We have checked by overlapping their plot with the $f_3$ plot generated using our definition, and confirmed that the HKS result agrees with ours.} However, their argument relies on the iterative estimate of the spectral density and the partition function. In this work, we only need the estimate on the partition function itself.
	\item (b) The proof of lemma \ref{lemma:iteration} also provides an iterative algorithm to determine $N$, $x_i$'s and $y_i$'s in lemma \ref{lemma:functionbound}. The algorithm is already summarized in figure \ref{fig:algorithmxy}. The detailed rules are as follows:
	\begin{itemize}
		\item $(x,y)\in\Omega_\alpha^{(0)}$. 
		
		For $x,y>1$, we choose $N=1$, $x_1=x$ and $y_1=y$.
		\item $(x,y)\in\Omega_\alpha^{(1)}\backslash\Omega_\alpha^{(0)}$
		
		If $y>1$ and $x<1$, we choose
		\begin{equation}
			\begin{split}
				N=2,\quad x_1=x,\ y_1=y,\quad x_2=\frac{1}{x},\ y_2=\frac{1}{y'},
			\end{split}
		\end{equation}
		where $x$, $y$ and $y'\in\left(0,1\right)$ satisfies \eqref{def:kappa}. If $x>1$ and $y<1$, the rule is similar.
		
		For $y>1$ and $x=1$, we need to let $x$ decrease a little bit but still larger than $f_1(y)$. A proper choice is $x\rightarrow x'=\frac{2}{\alpha(y-1)+2}$. Then we choose
		\begin{equation}
			\begin{split}
				N=2,\quad x_1=x,\ y_1=y,\quad x_2=\frac{1}{x'},\ y_2=\frac{1}{y'},
			\end{split}
		\end{equation}
		where $x'$, $y$ and $y'\in\left(0,1\right)$ satisfies \eqref{def:kappa} (with $x$ there replaced by $x'$). If $x>1$ and $y=1$, the rule is similar. 
		\item $k\geqslant1$, $(x,y)\in\Omega_\alpha^{(k+1)}\backslash\Omega_\alpha^{(k)}$
		
		For $k\geqslant1$, $y>1$ and $f_{k+1}(y)<x\leqslant f_k(y)$,  $N$ and other $x_i$'s and $y_i$'s are defined as follows. We choose $y'\in\left(0,1\right)$ such that $x$, $y$ and $y'$ solves \eqref{def:kappa}. Then we get a collection of points\footnote{Here we use the notation $\tilde{x}_i$ and $\tilde{y}_i$ for the points determined by $\left(\frac{1}{x},\frac{1}{y'}\right)$, as the notation $x_i$ and $y_i$ is reserved for the points determined by $(x, y)$.}
		$$(\tilde{x}_1,\tilde{y}_1)=\left(\frac{1}{x},\frac{1}{y'}\right),\ (\tilde{x}_2,\tilde{y}_2),\ldots,(\tilde{x}_{N},\tilde{y}_{N})$$
		determined inductively by $\left(\frac{1}{x},\frac{1}{y'}\right)$. We choose
		\begin{equation}
			\begin{split}
				x_1&=x,\ y_1=y, \\
				 x_i&=\tilde{x}_{i-1},\ y_i=\tilde{y}_{i-1}\quad (i=2,3,\ldots, N+1).
			\end{split}
		\end{equation}
		Similar rule for $x>1$ and $f_{k+1}(x)< y\leqslant f_k(x)$.
	\end{itemize}
\end{remark}

\section{Virasoro-primary partition function}\label{section:vir}
In the previous sections, our discussion did not fully use the Virasoro symmetry: the argument mainly relied on the diagonalizability of the Virasoro generators \(L_0\) and \(\bar{L}_0\). In this section, we provide comments on how to strengthen our results by taking advantage of the full Virasoro symmetry.

\subsection{Setup using Virasoro symmetry}
Let us consider a CFT with central charge \(c > 1\). We assume that the CFT Hilbert space is a direct sum of representations of the Virasoro algebra:
\begin{equation}\label{H:VirDecomp}
	\mathcal{H}_{\rm CFT}=\bigoplus_{h,\bar{h}} {\rm Vir}_{h}\otimes {\rm Vir}_{\bar{h}},
\end{equation}
where \({\rm Vir}_{h}\otimes {\rm Vir}_{\bar{h}}\) denotes the irreducible representation of the Virasoro algebra with highest conformal weights \((L_0, \bar{L}_0) = (h, \bar{h})\).

Using \eqref{H:VirDecomp}, and recalling the definition of the torus partition function in \eqref{def:Z}, we can express \(Z(\beta_{L}, \beta_{R})\) as a sum of Virasoro characters \(\chi_h(\beta_L)\chi_{\bar{h}}(\beta_R)\) over primaries:
\begin{equation}\label{Z:charexp}
	Z(\beta_L,\beta_R)=\sum\limits_{h,\bar{h}}n_{h,\bar{h}}\ \chi_h(\beta_L)\chi_{\bar{h}}(\beta_R),
\end{equation}
where \(n_{h,\bar{h}}\) counts the degeneracy of Virasoro primaries with conformal weights \(h\) and \(\bar{h}\).

For \(c > 1\), the characters of unitary Virasoro representations are given by:
\begin{equation}\label{def:Vircharacter}
	\chi_h (\beta) \equiv \text{Tr}_{{\rm Vir}_h} \left( e^{- \beta \left( L_0 -
		\frac{c}{24} \right)}  \right) = \frac{e^{\frac{c - 1}{24} \beta}}{\eta
		(\beta)} \times \begin{cases}
		1 - e^{- \beta} &\text{if } h = 0,\\
		e^{- h \beta }&\text{if } h > 0,
	\end{cases}
\end{equation}
where the Dedekind eta function \(\eta(\beta) \equiv e^{-\beta/24}\prod\limits_{n=1}^{\infty}(1-e^{-n \beta})\) accounts for the contribution of descendants. The factor \((1 - e^{-\beta})\) for \(h = 0\) arises because the level-1 null descendant and its descendants are subtracted.

The above setup allows us to define the Virasoro-primary partition function:
\begin{equation}\label{def:Zvir}
	Z^{\rm Vir}(\beta_{L},\beta_{R}) := Z(\beta_{L},\beta_{R}) \eta(\beta_{L}) \eta(\beta_R).
\end{equation}
Using eqs.\,\eqref{Z:charexp}, \eqref{def:Vircharacter}, and \eqref{def:Zvir}, we obtain the following expansion for the Virasoro-primary partition function:
\begin{equation}
	\begin{split}
		Z^{\rm Vir}(\beta_{L},\beta_{R}) = e^{\frac{c - 1}{24}(\beta_L + \beta_R)} \Bigg{[}&(1 - e^{- \beta_L}) (1 - e^{- \beta_R})+ \sum_{h, \bar{h}>0} n_{h,\bar{h}}\,e^{-h \beta_L  - \bar{h}\beta_R } \\
		&+ \sum\limits_{J=1}^{\infty} n_{J,0}\, e^{-J\beta_{L}} (1 - e^{- \beta_R}) 
		+ \sum\limits_{J=1}^{\infty} n_{0,J}\, (1 - e^{- \beta_L}) e^{-J \beta_{R}} \Bigg{]}. \\
	\end{split}
\end{equation}
The first term in the square brackets is the vacuum contribution. Here, \(n_{h,\bar{h}}\) counts the degeneracy of Virasoro primaries with conformal weights \((h, \bar{h})\), classified as follows:
\begin{equation}
	\begin{split}
		&{\rm ordinary\ weights}: \quad h, \bar{h} \neq 0, \\
		&{\rm spin-}J\ {\rm left\ currents}: \quad h = J,\ \bar{h} = 0, \\
		&{\rm spin-}J\ {\rm right\ currents}: \quad h = 0,\ \bar{h} = J. \\
	\end{split}
\end{equation}
Using the modular invariance condition \eqref{modularinv} and the modular property of the Dedekind eta function, \(\eta(\beta) = \sqrt{\frac{2\pi}{\beta}} \eta\left(\frac{4\pi^2}{\beta}\right)\), we obtain the modular invariance condition for the Virasoro-primary partition function:
\begin{equation}\label{modularinv:Vir}
	Z^{\rm Vir}(\beta_{L},\beta_{R}) = \sqrt{\frac{4\pi^2}{\beta_{L}\beta_{R}}}\, Z^{\rm Vir}\left(\frac{4\pi^2}{\beta_{L}},\frac{4\pi^2}{\beta_{R}}\right).
\end{equation}

\subsection{Generalization of theorem \ref{theorem:logZ}}
Similarly to the statement in theorem \ref{theorem:logZ}, we fix $\alpha \in (0, 1]$ and define the low- and high-twist parts of the Virasoro-primary partition function, with the factor of $e^{\frac{c-1}{24}}$ factorized out, as follows:
\begin{equation}\label{def:ZVirtildeLH}
	\begin{split}
		\tilde{Z}^{\rm Vir}_L(\alpha; \beta_{L}, \beta_{R}) & := (1 - e^{- \beta_L}) (1 - e^{- \beta_R}) + \sum_{0 < \text{min}(h, \bar{h}) < \frac{\alpha (c - 1)}{24}} n_{h, \bar{h}}\,e^{- h \beta_L - \bar{h} \beta_R } \\
		& \quad + \sum\limits_{J = 1}^{\infty} n_{J, 0} e^{- J \beta_L} (1 - e^{- \beta_R}) + \sum\limits_{J = 1}^{\infty} n_{0, J} (1 - e^{- \beta_L}) e^{- J \beta_R}, \\
		\tilde{Z}^{\rm Vir}_H(\alpha; \beta_{L}, \beta_{R}) & := \sum\limits_{h, \bar{h} \geqslant \frac{\alpha (c - 1)}{24}} n_{h, \bar{h}}\,e^{- h \beta_L - \bar{h} \beta_R}. \\
	\end{split}
\end{equation}
We define the error term for the Virasoro-primary free energy as follows:
\begin{equation}\label{def:errorVir}
	\mathcal{E}^{\rm Vir}(\beta_L, \beta_R) := \log Z^{\rm Vir}(\beta_L, \beta_R) - \frac{c - 1}{24} (\beta_L + \beta_R).
\end{equation}
Using the same argument as in theorem \ref{theorem:logZ}, we obtain the following result for the Virasoro-primary partition function:

\begin{theorem}\label{theorem:logZVir}
	Let $\alpha \in (0, 1]$ be fixed. Then for any $c > 1$ unitary, Virasoro invariant, modular invariant 2D CFT, the error term $\mathcal{E}^{\rm Vir}(\beta_L, \beta_R)$ of the Virasoro-primary free energy $\log Z^{\rm Vir}(\beta_L, \beta_R)$ (as defined in \eqref{def:errorVir}) is bounded from above by
	\begin{equation}\label{eq:Virerrorbound}
		\begin{split}
			\mathcal{E}^{\rm Vir}(\beta_L, \beta_R) \leqslant \log \Bigg{[} &\sum_{i = 1}^{N} \left( \prod\limits_{j = 2}^{i} \sqrt{\frac{\beta_L^{(j)} \beta_R^{(j)}}{4 \pi^2}} \right) \tilde{Z}^{\rm Vir}_L\left( \alpha; \beta_L^{(i)}, \beta_R^{(i)} \right) \\
			& + \left( \prod\limits_{j = 2}^{N} \sqrt{\frac{\beta_L^{(j)} \beta_R^{(j)}}{4 \pi^2}} \right) \tilde{Z}^{\rm Vir}_H\left( \alpha; \beta_L^{(N)}, \beta_R^{(N)} \right) \Bigg{]}, \quad (\beta_L, \beta_R) \in \mathcal{D}_\alpha.
		\end{split}
	\end{equation}
	In inequality \eqref{eq:Virerrorbound}, $\tilde{Z}^{\rm Vir}_L$ and $\tilde{Z}^{\rm Vir}_H$ are defined in \eqref{def:ZVirtildeLH}; the domain $\mathcal{D}_\alpha$ is defined in \eqref{def:Dalpha}; the numbers $N$, $\beta_L^{(i)}$'s, and $\beta_R^{(i)}$'s are defined by the algorithm in figure \ref{fig:algorithmbeta}.
\end{theorem}

We omit the proof of theorem \ref{theorem:logZVir} since almost all the technical details are the same as the proof of theorem \ref{theorem:logZ}. The only difference is that for the Virasoro-primary partition function, the modular invariance condition \eqref{modularinv:Vir} includes an additional factor of $\sqrt{\frac{4 \pi^2}{\beta_L \beta_R}}$, compared to the modular invariance condition \eqref{modularinv} for the full partition function. This extra factor leads to the additional factors of $\sqrt{\frac{\beta_L^{(j)} \beta_R^{(j)}}{4 \pi^2}}$ on the right-hand side of inequality \eqref{eq:Virerrorbound}.

\subsection{Strengthened large-$c$ universality}
We now consider a unitary, Virasoro invariant, modular invariant 2D CFT that satisfies the properties listed in section \ref{subsec:HKS}, but with weaker sparseness conditions:
\begin{itemize}
	\item The theory has a normalizable vacuum.
	\item The central charge of the theory is tunable and allows the limit \(c \rightarrow \infty\).
	\item There exists a fixed \(\epsilon > 0\), and the spectrum of Virasoro primaries of scaling dimensions \(\Delta := h + \bar{h}\) below \(\frac{c}{12} + \epsilon\) is sparse, such that:
	\begin{equation}\label{Vir:HKSenergy}
		\widehat{Z}^{\rm Vir}_L(\beta) := \sum\limits_{h + \bar{h} \leqslant \frac{c}{12} + \epsilon} n_{h, \bar{h}} \ e^{-(h + \bar{h}) \beta} \leqslant A(\beta), \quad \beta > 2\pi,
	\end{equation}
	where \(A(\beta) < \infty\) for \(\beta > 2\pi\), and \(A\) does not depend on \(c\).
	
	\item The spectrum of Virasoro primaries of twist (defined by \(\min\left\{2h, 2\bar{h}\right\}\)) below \(\frac{c-1}{12}\) is sparse, such that:
	\begin{equation}\label{Vir:HKStwist}
		\tilde{Z}^{\rm Vir}_L(\alpha; \beta_L, \beta_R) \leqslant B(\beta_L, \beta_R), \quad \beta_L \beta_R > 4\pi^2,
	\end{equation}
	where $\tilde{Z}^{\rm Vir}_L(\alpha; \beta_L, \beta_R)$ is defined in \eqref{def:ZVirtildeLH}, \(B(\beta_L, \beta_R) < \infty\) for \(\beta_L \beta_R > 4\pi^2\), and \(B\) does not depend on \(c\).
\end{itemize}
Under these assumptions, we would like to show that the conclusion of corollary \ref{cor:alpha} still holds, despite the weaker sparseness conditions.

Using \eqref{FE:lowT}, \eqref{def:Zvir}, and \eqref{def:errorVir}, we rewrite the error term $\mathcal{E}(\beta_L, \beta_R)$ of the free energy $\log Z$ as follows:
\begin{equation}\label{error:rewrite}
	\begin{split}
		\mathcal{E}(\beta_L, \beta_R) = \mathcal{E}^{\rm Vir}(\beta_L, \beta_R) - \frac{1}{24}(\beta_L + \beta_R) + \log \eta(\beta_L) + \log \eta(\beta_R).
	\end{split}
\end{equation}
By theorem \ref{theorem:logZVir} and the sparseness conditions above, $\mathcal{E}^{\rm Vir}(\beta_L, \beta_R)$ is of order one in $c$ when $(\beta_L, \beta_R) \in \mathcal{D}_\alpha$, where $\mathcal{D}_\alpha$ is given in \eqref{def:Dalpha}. Additionally, the other terms on the right-hand side of \eqref{error:rewrite} do not depend on $c$. Therefore, the error term $\mathcal{E}(\beta_L, \beta_R)$ is also of order one in $c$ when $(\beta_L, \beta_R) \in \mathcal{D}_\alpha$. By modular invariance \eqref{modularinv}, we conclude:
\begin{corollary}\label{cor:Virlargec}
	Under the above assumptions, the CFT free energy satisfies the following asymptotic behavior in the limit \(c \rightarrow \infty\) with \(\beta_L\) and \(\beta_R\) fixed:
	\begin{equation}\label{eq:Virlargec}
		\begin{split}
			\log Z(\beta_L, \beta_R) = \begin{cases}
				\frac{c}{24} (\beta_L + \beta_R) + {\rm O}(1), & (\beta_L, \beta_R) \in \mathcal{D}_\alpha, \\
				\frac{\pi^2 c}{6} \left( \frac{1}{\beta_L} + \frac{1}{\beta_R} \right) + {\rm O}(1), & \left( \frac{4\pi^2}{\beta_L}, \frac{4\pi^2}{\beta_R} \right) \in \mathcal{D}_\alpha. \\
			\end{cases}
		\end{split}
	\end{equation}
	Here, by {\rm O(1)} we mean that the error term \(\mathcal{E}(\beta_L, \beta_R)\) is of order one in $c$, i.e., it is bounded in absolute value by some finite number that only depends on $\beta_L$ and $\beta_R$.
\end{corollary}
A special case of corollary \ref{cor:Virlargec} is when $\alpha = 1$. In this case, we observe that the sparseness condition on the twist is slightly weaker than the original one proposed by HKS in \cite{Hartman:2014oaa}: here, we only assume sparseness for \(\min(h, \bar{h}) < \frac{c - 1}{24}\), instead of \(\frac{c}{24}\).

This shift of \(\frac{1}{24}\) in the sparseness condition is particularly interesting because it may avoid the necessity of knowing an infinite amount of data.\footnote{We thank Tom Hartman for bringing this to our attention.} It is known that a 2D CFT with \(c > 1\) and a twist gap must contain infinitely many heavy states near the line \(h = \frac{c - 1}{24}\) and the line \(\bar{h} = \frac{c - 1}{24}\) in the \((h, \bar{h})\)-plane \cite{Collier:2016cls,Afkhami-Jeddi:2017idc,Benjamin:2019stq,Pal:2022vqc}. Furthermore, in some partition functions, the spectrum of Virasoro primaries below \(\frac{c - 1}{24}\) may be extremely sparse. For example, one can modify the MWK partition function \cite{Maloney:2007ud, Keller:2014xba} by adding conical-defect states \cite{Benjamin:2020mfz} or stringy states \cite{DiUbaldo:2023hkc}, along with their PSL$(2; \mathbb{Z})$ images, to construct a unitary, modular invariant partition function. For these examples, applying corollary \ref{cor:alpha} may be challenging due to the presence of infinitely many states with \(\min(h, \bar{h})\) between \(\frac{c - 1}{24}\) and \(\frac{c}{24}\). However, corollary \ref{cor:Virlargec} (with $\alpha=1$) applies directly since there are only finitely many primary states with \(\min(h, \bar{h})\) below \(\frac{c - 1}{24}\).

\section{Conclusion}\label{sec:conclusion}
In this paper,  we proved a bound on the vacuum-subtracted free energy:
$$\mathcal{E}(\beta_L,\beta_R)=\log Z(\beta_{L},\beta_{R})-\frac{c}{24}(\beta_{L}+\beta_{R})\,,$$ 
for $\beta_L,\beta_R\in\mathcal{D}_\alpha \subseteq \{(\beta_L,\beta_R): \beta_L\beta_R>4\pi^2\}$ in terms of the low lying spectrum,\footnote{$\tilde{Z}_H$ appearing in the theorem \ref{theorem:logZ} can easily be bounded by $\widehat{Z}_{L}(\beta)$ using an argument from \cite{Hartman:2014oaa},  see \eqref{modifiedTh}.}  in particular,  
\begin{equation*}
	\begin{split}
		\tilde{Z}_L(\alpha;\beta_{L},\beta_{R}):=\sum\limits_{\text{min}(h,\bar{h})\leqslant\frac{\alpha c}{24}}n_{h,\bar{h}}\ e^{- h\beta_L-\bar{h}\beta_{R}}\,,\quad 
		\widehat{Z}_{L}(\beta):=\sum_{\Delta\leqslant c/12+\epsilon} e^{-\beta\Delta}\,.
	\end{split}
\end{equation*}
The bound is summarized in theorem \ref{theorem:logZ}. The domain $\mathcal{D}_\alpha$ is given by \eqref{def:Dalpha}.   A similar bound on $\mathcal{E}(\beta_L,\beta_R)$ for $\beta_L,\beta_R\in\mathcal{D}'_\alpha$,  where $\mathcal{D}'_\alpha$ is the image of $S$ modular transformation of $\mathcal{D}_\alpha$:
 \begin{equation}
	\begin{split}
	\mathcal{D}'_\alpha=\left\{(\beta_{L},\beta_{R})\,\Big{|}\, \left(\frac{4\pi^2}{\beta_{L}},\frac{4\pi^2}{\beta_{R}}\right)\in\mathcal{D}_\alpha\right\} \subseteq \{(\beta_L,\beta_R): \beta_L\beta_R<4\pi^2\}\,,
	\end{split}
\end{equation}
follows from the modular invariance of the partition function.  The domain $\mathcal{D}_\alpha$ is obtained using an iteration equation for the domain of validity of the bound on the $(\beta_L,\beta_R)$ plane and solving the limit of the iteration.  

As a consequence of our bound,  we derived the phase diagram of the free energy for CFTs with large central charge,  assuming sparseness in the low lying energy ($E<\epsilon$) and twist spectra $(\tau \leqslant \frac{\alpha c}{12}$),  as a function of $\beta_L$ and $\beta_R$.  The universal domain depends on the parameter $\alpha$, which quantifies the amount of sparseness in the twist. For $0<\alpha\leqslant1$,  we showed that for $(\beta_L,\beta_R)$ in the union of two disconnected regimes $\mathcal{D}_\alpha\sqcup\mathcal{D}'_\alpha$, the large-$c$ free energy has the universal asymptotic behavior, demonstrated in corollary \ref{cor:alpha}. 

Note that $\mathcal{D}_1\sqcup\mathcal{D}'_1$ covers the whole plane except the hyperbola $\beta_L\beta_R=4\pi^2$.  This amounts to proving the HKS conjecture \ref{conjecture:HKS}.  Using holography,  we deduce that with $\alpha=1$ sparseness, the black hole saddle dominates if $\beta_L\beta_R<4\pi^2$ while thermal AdS$_3$ dominates if $\beta_L\beta_R>4\pi^2$.  For $\alpha<1$,  $\mathcal{D}_\alpha\sqcup\mathcal{D}'_\alpha$ does not cover the whole plane,  see figure \ref{fig:1} and \ref{fig:2}.  This is intuitively expected since we are assuming the weaker sparseness condition $\alpha<1$.  The weaker sparseness condition allows us to have more low lying states,  leading to non-universal regimes on the $(\beta_L,\beta_R)$ plane, where neither the black hole nor the thermal AdS is known to dominate the free energy in the semiclassical limit. 

Theorem \ref{theorem:logZ} and its applications do not use Virasoro symmetry. Using Virasoro symmetry, we extended the results to a refined version, summarized in theorem \ref{theorem:logZVir}, which concerns the Virasoro-primary partition function that counts only Virasoro primaries. This extension allows us to relax the twist sparseness condition: the Virasoro primaries with twist below $\frac{\alpha(c-1)}{12}$ (in contrast to all states with twist below $\frac{\alpha c}{12}$) are sparse, yet we can still prove the large-$c$ universal behavior of the free energy (corollary \ref{cor:Virlargec}), particularly the HKS conjecture. This result is advantageous as it potentially enables us to avoid imposing a sparseness condition on an infinite amount of CFT data. Notably, there are infinitely many heavy states near the line \(h = \frac{c - 1}{24}\) and the line \(\bar{h} = \frac{c - 1}{24}\) for a 2D CFT with \(c > 1\) and a non-zero twist gap \cite{Collier:2016cls,Afkhami-Jeddi:2017idc,Benjamin:2019stq,Pal:2022vqc}.


A testbed for verifying these predictions, in particular, the $\alpha=1$ case, is symmetric product orbifolds in $N\to\infty$ limit,  which appear widely in the counting of black hole microstates \cite{Strominger:1996sh, Sen:2007qy,  Sen:2012cj}.  Indeed the free energy is verified to be universal as a fucntion of $\beta_L,\beta_R$ in such CFTs in $N\to\infty$ limit \cite{Keller:2011xi,Hartman:2014oaa}.  It would be interesting to have an example where the sparseness condition is satisfied with $\alpha<1$ and verify our predictions for $\alpha<1$.

We end the concluding section by making some remarks about near extremal black holes and laying out the potential utility of our methods in a wider context.\\

\noindent$\bullet$ \textbf{Comments about near extremal black holes}\\

 As our result predicts the regime where the black hole dominates the grand-canonical ensemble,  it is meaningful to ask whether we can make any statement about near extremal black holes,  which has garnered lot of attentions in recent years.  In particular, 
the Schwarzian/JT gravity theory has been used to study near-extremal black holes \cite{Almheiri:2014cka,Almheiri:2016fws,Mertens:2017mtv,Mertens:2022irh,Kitaev:2018wpr,Yang:2018gdb,Nayak:2018qej,Moitra:2018jqs,Moitra:2019bub,Castro:2018ffi,Castro:2019crn,Hernandez-Cuenca:2024icn}; see \cite{Ghosh:2019rcj} for a perspective from $2$D CFT, and \cite{Pal:2023cgk} for its rigorous avatar. For an analysis without dimensional reduction to JT gravity, see \cite{Rakic:2023vhv,Kolanowski:2024zrq}, and for a calculation in the stringy regime, refer to \cite{Ferko:2024uxi}.  A concrete prediction of our result in this context can be summarized as follows:

\textit{We consider a putative quantum gravitational path integral defining a grand-canonical ensemble with parameters $\beta_R,\beta_L$.  In the $c\to\infty$ limit,  we can trust semiclassical weakly coupled Einstein gravity.   If we make sure that $\beta_L\beta_R<4\pi^2$ with $\beta_R=O(c^0)$, $\beta_L=O(c^0)$   the dominating saddle is given by a rotating black hole.  Now if we fix $\beta_R=O(c^0)$ and keep increasing $\beta_L$ to approach extremality,  owing to the non-perturbative corrections in $c$ (as estimated by performing a CFT computation),  the black hole ceases to dominate at some point,  well before the  one-loop correction ($1/c$)  coming from the Schwarzian action has a chance to become important i.e.  for $\beta_R=O(c^0)$,  the one-loop correction ($1/c$)  coming from the Schwarzian action always remains subdominant in the grandcanonical ensemble.}


To elaborate on the above,  recall that with $\alpha=1$ sparseness, the black hole saddle dominates if $\beta_L\beta_R<4\pi^2$ while thermal AdS$_3$ dominates if $\beta_L\beta_R>4\pi^2$. When $\beta_L$ is very large ($1/c\ll\beta_R\ll1\ll\beta_L\ll c$) yet $\beta_L\beta_R<4\pi^2$, the corresponding black hole saddle is nothing but near extremal BTZ black hole.\footnote{Let us clarify here why we restrict ourselves to the regime $1/c \ll \beta_R \ll 1 \ll \beta_L \ll c$. First, for our analysis to be valid, we need both $\beta_{L}\beta_{R} < 4\pi^2$ (for the black hole dominance) and $1/\beta_{L}, 1/\beta_{R} \ll c$ (for the large $c$ limit). Then we want an extremal BTZ black hole, which means that either $\beta_{L} \ll 1 \ll \beta_{R}$ or $\beta_{R} \ll 1 \ll \beta_{L}$. Consider the case $\beta_{R} \ll 1 \ll \beta_{L}$, then we have $\beta_{L} < \frac{4\pi^2}{\beta_{R}} \ll c$. This gives $1/c \ll \beta_R \ll 1 \ll \beta_L \ll c$.
} 
Note that in the near extremal limit,  $\beta \sim \beta_L/2$. 
We note that the $1/c$ correction from Schwarzian naively has a chance to become important\footnote{In \cite{Ghosh:2019rcj} ,  the authors consider $\beta_L\sim c$ and $\beta_R\sim 1/c$ to make the Schwarzian effect on a same footing as the leading answer.  In particular,  under the dimensional reduction of $3$D gravity,  the Schwarzian approximation should give a contribution of the form $\frac{\pi^2c}{6\beta_L}- \frac{3}{2}\log (\beta_L/c)$ and consequently $\beta_L\sim c$ put them on same footing.  However,  there should be an extra factor of $3/2\log c$ coming from comparing the $2$D path integral with that of $3$D one,  making the full answer $\frac{\pi^2c}{6\beta_L}- \frac{3}{2}\log (\beta_L)$,  as expected from $2$D CFT.  So the one-loop effect becomes important when $\beta_L\log\beta_L\sim c$. We thank Joaquin Turiaci for making the suggestion of possible presence of the $\log c$ factor coming from comparing the path integral measure. } when $\beta_L\log\beta_L\sim c$,  however,  this regime of parameter definitely violates the bound $\beta_L\beta_R<4\pi^2$ if $\beta_R=O(c^0)$ is kept fixed.
This follows from 3D gravity saddle-point computation, and is consistent with the statement of HKS conjecture \ref{conjecture:HKS}  that we prove here.  As mentioned earlier,  in the limit that we are studying, one-loop correction coming from the Schwarzian i.e.\ the $\log \beta$ term is supposed to be a subleading one and one cannot cleanly separate this effect from the order one contributions coming from the sparse low lying spectrum.  We leave a more careful treatment of this for future.  
%
We make a further remark that the presence of the lower bound on temperature naturally resolves the puzzle of  entropy turning negative in a UV complete theory.
See \cite{Engelhardt:2020qpv, Hernandez-Cuenca:2024icn} for a more general discussion beyond AdS$_3$/CFT$_2$ set up. 

To summarize,  in the (grand-)canonical ensemble,  only above the critical temperature,  the black hole dominates and it makes sense to discuss Schwarzian approximation.  It might still make sense to talk about near extremal black hole microstates in the micro-canonical ensemble,  for example,  see the appendix A of \cite{Hartman:2014oaa},  and discuss a regime of parameter where schwarzian correction might posssibly become important,  however we do not make any statement in the microcanonical ensemble in this paper. 



The correlation function of light operators in the double lightcone limit is expected to show similar universal features \cite{Kraus:2018pax}  and capture the physics of near extremal black holes \cite{Ghosh:2019rcj}.  It would be an useful pursuit to extend our rigorous analysis in the context of correlation functions. This requires a careful estimation of thermal and/or four point Virasoro blocks \cite{Collier:2018exn,Kusuki:2018wpa,Eberhardt:2023mrq} in the appropriate limit.

 Recently it has been pointed out that in the higher dimensional AdS,  due to superradiance,  non-supersymmetric near extremal black holes do not dominate the canonical ensemble.  Rather the dominating saddle is \textit{grey galaxy} \cite{Kim:2023sig}.  Thus a more appropriate set up to study the near extremal black holes is supersymmetric one \cite{Choi:2024xnv}.  It would be nice to extend our analysis to the supersymmetric set up,  at least in the context of AdS$_3$/CFT$_2$.  It is straightforward to work out the $\mathcal{N}=1$ case,  which is  therefore left as an excercise to the readers.  \\

\noindent$\bullet$ \textbf{Potential utility in wide context.}\\

We envision our results to be useful in a wider context.  We mention some of them here.  The references \cite{Benjamin:2015hsa} and \cite{Benjamin:2015vkc} extended the ideas of HKS to weak Jacobi forms by examining the elliptic genus of CFT$_2$.  They compared the asymptotic growth of the fourier coefficients of this genus to the entropy of BPS black holes in AdS$_3$ supergravity. Further studies of these coefficients appeared in \cite{Belin:2019rba, Belin:2020nmp,Benjamin:2022jin} in the context of charting the space of holographic CFTs.  The weak Jacobi forms with appropriate sparseness condition have recently been studied in \cite{Apolo:2024kxn},  where the authors identified the logarithmic correction to the leading behavior of the asymptotic growth of the fourier coefficients.  Furthermore,  the techniques of HKS have proven to be useful in the context of ensemble of Narain CFTs \cite{Dymarsky:2020pzc}, and in making the Cardy formula mathematically precise \cite{Mukhametzhanov:2019pzy,Mukhametzhanov:2020swe,Pal:2019zzr,Das:2020uax}.  Furthermore,  it is conceivable to extend our analysis to CFTs with global symmetry. We expect that the free energy computed within each charged sector will exhibit the same kind of universality as described in the HKS framework.

\section*{Acknowledgements}
We are grateful to Nathan Benjamin, Anatoly Dymarsky, Victor Gorbenko, Shiraz Minwalla, Mukund Rangamani, Ashoke Sen, Sandip Trivedi and Joaquin Turiaci for insightful discussions. We especially thank Tom Hartman, Mrunmay Jagadale, Hirosi Ooguri, Eric Perlmutter, Slava Rychkov and Allic Sivaramakrishnan for their valuable feedback on the draft. JQ would also like to thank Queen Mary University of London for their hospitality during the final stages of the draft preparation. ID acknowledges support from the Government of India, Department of Atomic Energy, under Project Identification No. RTI 4002, and from the Quantum Space-Time Endowment of the Infosys Science Foundation. SP is supported by the U.S. Department of Energy, Office of Science, Office of High Energy Physics, under Award Number DE-SC0011632, and by the Walter Burke Institute for Theoretical Physics. JQ is supported by Simons Foundation grant 994310 (Simons Collaboration on Confinement and QCD Strings).

\small
	
\bibliography{largec}
\bibliographystyle{utphys}
	
\end{document}